\documentclass[11pt]{amsart}

\usepackage{macros-master,amsaddr,physics}

\renewcommand{\op}{\operatorname}

\newcommand{\assign}{:=}
\newcommand{\nin}{\not\in}
\newcommand{\nobracket}{}
\newcommand{\tmdummy}{$\mbox{}$}
\newcommand{\tmmathbf}[1]{\ensuremath{\boldsymbol{#1}}}
\newcommand{\tmop}[1]{\ensuremath{\operatorname{#1}}}
\newcommand{\tmstrong}[1]{\textbf{#1}}

\newenvironment{enumeratenumeric}{\begin{enumerate}[1.] }{\end{enumerate}}
\newenvironment{enumerateroman}{\begin{enumerate}[i.] }{\end{enumerate}}

\begin{document}

\title{On the renormalization and quantization of topological--holomorphic field theories}
\author{Minghao Wang and Brian R. Williams}
\thanks{Boston University, Department of Mathematics and Statistics}
\email{minghaow@bu.edu}
\email{bwill22@bu.edu}
\begin{abstract}
    Topological field theories and holomorphic field theories naturally appear in both mathematics and physics. However, there exist intriguing hybrid theories that are topological in some directions and holomorphic in others, such as twists of supersymmetric field theories or Costello's 4-dimensional Chern-Simons theory. In this paper, we rigorously prove the ultraviolet (UV) finiteness for such hybrid theories on the model manifold $\R^{d'} \times \C^d$, and present two significant vanishing results regarding anomalies: in the case $d'=1$, the odd-loop obstructions to quantization on $\R^{d'} \times \C^d$ vanish; in the case $d'>1$, all obstructions disappear, allowing us to define a factorization algebra structure for quantum observables.
\end{abstract}
\maketitle

\section{Introduction}

This work continues the study of quantum field theories which exist on a product of manifolds of the form $M \times X$ where $M$ is a smooth manifold and $X$ is a complex manifold \cite{GRWthf,Gaiotto:2024gii}.
Such theories depend only on the smooth structure of the manifold $M$ and only on the complex structure of the manifold $X$.
In this sense, these quantum field theories generalize the well-studied \textit{topological} field theories and \textit{holomorphic} field theories which have garnered significant attention in recent years \cite{bershadsky1994kodaira, Kapustin2005ChiralDR, costello2015quantization, MR3126501,BWhol,wang2024feynman}.
Examples of these hybrid theories appear as twists of supersymmetric theories, see for example \cite{kapustin2006holomorphic,Oh:2019mcg,Garner:2023wrc,Aganagic:2017tvx}.
We recall a precise definition of a topological-holomorphic field theory in section \ref{sec:dfn}.

The main result of this paper is that the perturbative quantization, on flat space $\R^{d'} \times \C^d$, of classical field theories which are of this hybrid topological and holomorphic type are particularly well-behaved.
In particular, we will show that in many cases, quantizations always exist and can be produced diagrammatically using Feynman graph integrals.
We refer to theorem \ref{thm:renormalization} for a more detailed statement of this result.

\begin{thm}\label{main theorem1}
The perturbative renormalization of a topological-holomorphic theory on $\R^{d'} \times \C^d$ is \textit{UV finite} to all orders in perturbation theory.
\end{thm}

This theorem extends a result of the second author along with Gwilliam and Rabinovich to \textit{all} orders in perturbation theory, rather than just to first-order \cite{GRWthf}. The main technique used here is the compactification of the space of Schwinger parameters, developed by the first author \cite{wang2024feynman}. 
Most of section \ref{sec:feynman} is dedicated to the proof of this result.

In the remainder of this introduction we detail a consequence of our main result which is a rigorous construction of \textit{factorization algebras} associated to a wide variety of such topological-holomorphic theories.

Factorization algebras (and their cousins, chiral algebras) originally appeared in the work of Beilinson and Drinfeld \cite{BD} in their approach to codify algebraic structures in conformal field theory.
More generally, as developed by Costello and Gwilliam, a factorization algebra encapsulates the observables in \textit{any} quantum field theory \cite{CG2}.
To every open set $U \subset M$, the factorization algebra of observables $\Obs$ assigns a complex vector space (this becomes a cochain complex in the Batalin--Vilkovisky formalism) that we denote $\Obs(U)$; this is the set of measurements that one can make in the patch of spacetime $U$.
If $U \subset V \subset M$ is a nested sequence of open sets there is a extension map $\Obs(U) \to \Obs(V)$; this means that observables on $U$ can be viewed as an observable on $V$.
The most important structure arises from combining, or multiplying, observables.
If $U,U'$ are disjoint open sets which sit inside a larger open set $V$, then combining observables takes the form of a map $\star \colon \Obs(U) \otimes \Obs(U') \to \Obs(V)$.
Iterating this, one obtains higher arity multiplication maps when there are more than two disjoint open sets involved, and there is an associativity axiom which relates such multiplications.
Finally, there is a `locality' axiom which relates the value of the factorization algebra on $M$ (or some other open set) in terms of a cover of $M$ by open sets.
We refer to \cite{CG1} for details.

In practice, we obtain such factorization algebras using perturbative renormalization, which is very close in spirit to deformation quantization.
Suppose we have a classical field theory on a manifold $M$ prescribed by some action functional.
To every open set $U \subset M$, solutions to the Euler--Lagrange equations of the action functional restricted to $U$, taken up to gauge equivalences, defines a formal derived stack $\cE \cL(U)$. 
The factorization algebra of classical observables $\Obs^{cl}$ assigns to the open set $U$ the algebra of functions on this derived stack $\Obs^{cl}(U) = \cO(\cE \cL(U))$.

To study quantization we work in the formalism of quantum field theory developed in the series of books by Costello \cite{CostelloBook} and Costello, Gwilliam \cite{CG2}.
We give a more detailed recollection in section \ref{sec:dfn}.
Perturbatively, a quantum field theory is defined over the ring $\C[[\hbar]]$.
The problem is to find a factorization algebra $\Obs^{q}$ defined over $\C[[\hbar]]$ which reduces, modulo $\hbar$, to the factorization algebra of classical observables $\Obs^{cl}$.
Unlike deformation quantization, a complicating feature in field theory is that such quantizations may not exist.
This is encapsulated in the form of \textit{anomalies}. 
For a well-known instance of this we point out the Green--Schwarz mechanism which is used to cancel a one-loop anomaly in ten-dimensional supersymmetric Yang--Mills theory which can be canceled by carefully coupling to ten-dimensional supergravity.

We prove anomalies vanish when a theory has at least two topological directions. We refer to theorem \ref{anomaly free theorem1} for a more detailed statement.

\begin{thm}\label{main theorem2}
If $d' > 1$ then any classical perturbative topological-holomorphic theory defined on $\R^{d'} \times \C^d$ is anomaly free and admits a quantization to all orders in $\hbar$.

In particular, when $d'>1$ any topological-holomorphic theory defines a factorization algebra of quantum observables $\Obs^q$ on $\R^{d'} \times \C^d$, defined over the ring $\C[[\hbar]]$, which satisfies
\beqn
\Obs^q \otimes_{\C[[\hbar]]} \C \cong \Obs^{cl}
\eeqn
with $\Obs^{cl}$ the factorization algebra of classical observables.

When $d'=1$, any topological-holomorphic theory on $\R \times \C^d$ admits a quantization to \textbf{first-order} in perturbation theory, that is over  $\C[[\hbar]]\slash \hbar^2$.
\end{thm}

We refer proposition \ref{anomaly free theorem2} for details on anomalies in the cases $d' \geq 1$.
Indeed, we expect that topological-holomorphic Chern--Simons theory on $\R \times \C^2$ for a Lie algebra $\lie{g}$ suffers from a two-loop anomaly.

\begin{rmk}
There is a recent paper which shares some complementary overlap with our results \cite{Gaiotto:2024gii}.
However, there are two differences: 
\begin{enumerate}
    \item We prove the ultraviolet finiteness for all graphs, whereas \cite{Gaiotto:2024gii} focus only on an important class of graphs called Laman graphs. 
As far as we understand, finiteness for Laman graphs is not enough to prove the existence of the full quantum factorization algebra. 
However, by the conjectural formality for topological-holomorphic (colored) operads when $d'>1$, the Feynman graph integrals for Laman graphs may contain all the information of the factorization algebra.
    In the approach of \cite{Gaiotto:2024gii} an explicit model for this operadic structure in the form of $\lambda$-brackets and descent is given and only depends on the behavior of Laman graphs.
    \item In the proof of anomaly freeness in \cite{Gaiotto:2024gii}, they assumed an equivalence between topological gauge fixing and holomorphic gauge fixing, which doesn't seem obvious. We didn't assume such an equivalence.
\end{enumerate}
\end{rmk}

\subsection*{Acknowledgements}

We acknowledge Davide Gaiotto, Justin Kulp, and Jingxiang Wu for their work \cite{Gaiotto:2024gii} which partly inspired us to extend the perturbative Feynman diagram analysis of topological-holomorphic theories in order to produce the full factorization algebra of quantum observables on flat space. 
We especially thank Justin Kulp for a detailed reading of a draft and for helpful comments and feedback. 
We warmly thank Zhengping Gui for pointing out an error in the first version of the proof of Proposition \ref{vanishing of anomaly121}.
We also thank Si Li and Matt Szczesny for numerous conversations pertaining to this work.
Additionally, we would like to thank the three referees who's questions and suggestions prompted various improvements in presentation and organization.
Finally, we thank Boston University for support.

\subsection*{Conventions}
  {\tmdummy}
  
  \begin{enumeratenumeric}
    \item If $S_1, S_2, \ldots S_m$ are sets, $A_1, A_2, \ldots A_m$ are
    finite sets, an element in
    \[ \prod_{i = 1}^m (S_i)^{| A_i |} \]
    will be denoted by $\left( s_{a_1}, s_{a_2} {, \ldots, s_{a_m}} 
    \right)_{a_1 \in A_1, \ldots, a_m \in A_m}$, where $s_{a_i}$ is an element
    in $S_i$.
    
    \item Given a manifold $M$, we use $\Omega^{\bu} (M)$ to denote the space
    of differential forms on $M$. We use the natural $C^{\infty}$ topology on $\Omega^{\bu} (M)$, so $\Omega^{\bu} (M)$ is a Fréchet space. The space of distributional valued
    differential forms is denoted by $\mathcal{D}^{\bu} (M)$. If $M$ is a
    complex manifold, to emphasize their bi-graded nature, we will use
    $\Omega^{\bu, \bu} (M)$ and $\mathcal{D}^{\bu, \bu} (M)$ for
    differential forms and distributional valued differential forms
    respectively.

    \item We use $\hotimes$ to denote the complete projective tensor product of nuclear vector spaces. See \cite{treves2006topological} for more details about tensor products of nuclear spaces.
    
    \item We use the following Koszul sign rules: If $M, N$ are manifolds
    whose dimensions are $m, n$ respectively, \ $\alpha \in \Omega^i (M),
    \beta \in \Omega^j (N)$
    \begin{enumerateroman}
      \item Sign rule for differential forms:
      \[ \alpha \wedge \beta = (- 1)^{i j} \beta \wedge \alpha . \]
      \item Sign rule for integrations:
      \[ \int_{M \times N} \alpha \wedge \beta = \int_M \int_N \alpha \wedge
         \beta = (- 1)^{n i} \int_M \alpha \int_N \beta . \]
    \end{enumerateroman}
  \end{enumeratenumeric}

\section{Batalin--Vilkovisky formalism for topological-holomorphic theories}\label{sec:dfn}

We begin with an overview of the sort of classical field theories whose quantizations we consider in this paper.
For more details we refer to \cite[section 1]{GRWthf} \cite{Gaiotto:2024gii}.

\subsection{Classical topological-holomorphic theories}

Throughout we will only work on $(d' + 2d)$-dimensional space $\R^{d'} \times \C^d$.
The theories we consider are holomorphic `along' $\C^d$ and topological `along' $\R^{d'}$. 
In this short section we present such theories using action functionals, but the concept applies to theories without such a description.

For theories admitting a Lagrangian description, the action functional takes the form
\beqn\label{eqn:action}
S(\phi) = \int_{\R^{d'} \times \C^d} \d^d z \, \<\phi , (\dbar + \d_{deRham}) \phi\> + I(\phi)
\eeqn
where $\dbar$ is the $\dbar$-operator on $\C^d$ and $\d_{deRham}$ is the de Rham operator on $\R^{d'}$.\footnote{In \cite{GRWthf} we consider the slightly more general situation where the kinetic term could also contain $\int \phi Q^{hol} \phi$ where $Q^{hol}$ is some holomorphic differential operator. 
For our purposes, such a term can be absorbed in $I(\phi)$.}
The interaction $I(\phi)$ is given by a Lagrangian density which depends only on the holomorphic jets of fields (meaning no derivatives along $\R^{d'}$ appear). More precisely,
\begin{equation}\label{stronger assumption}
I(\phi)=\sum_{k=3}^{\infty}\sum_{\alpha_{1},\dots,\alpha_{k}\in\mathbf{N}}\int_{\R^{d'} \times \C^d} \d^d z \,f_{k,\alpha_{1},\dots,\alpha_{k}} \prod_{i=1}^{k}D^{\alpha_{i}}_{i}\phi,    
\end{equation}
where $D^{\alpha_{i}}_{i}$ is a differential operator with constant coeffients, and it only contains derivatives with respect to the holomorphic coordinates $\{z_{i}\}_{i=1}^{d}$. Here, $f_{k,\alpha_{1},\dots,\alpha_{k}}$ are smooth differential forms, and for each $k\geq3$, $f_{k,\alpha_{1},\dots,\alpha_{k}}=0$ for all but finitely many $\alpha_{1},\alpha_{2},\dots\alpha_{k}\in\mathbf{N}$.

We work within the Batalin--Vilkovisky formalism, meaning our space of fields is $\Z$-graded. 
In physics, this cohomological degree corresponds to the ghost number.
The graded space of fields of a topological-holomorphic theory with action functional as above is of the form
\beqn
\phi \in \Omega^{\bu}(\R^{d'}) \hotimes \Omega^{0,\bu}(\C^d) \otimes V .
\eeqn
Here, $V$ is a finite dimensional $\Z$-graded vector space equipped with a non-degenerate graded skew-symmetric pairing $\<-,-\>_V$ of cohomological degree $d+d'-1$.
This pairing induces the BV pairing on compactly supported fields by integration against the holomorphic volume element
\beqn\label{eqn:omega}
\omega(\phi, \phi') = \int_{\R^{d'} \times \C^d} \d^d z \, \<\phi,\phi'\>_V .
\eeqn
Dual to this pairing is the so-called BV (anti) bracket which is defined on (local) functionals of fields and is denoted $\{-,-\}$. See \cite[Lemma 3.2.3]{CostelloBook} for details.
In the BV formalism, the action functional of a classical field theory is required to satisfy the classical master equation:
\beqn
\{S,S\} = 0 .
\eeqn
For a topological-holomorphic theory as in \eqref{eqn:action}, this equation can equivalently be written as
\beqn
\d_{deRham} I + \dbar I + \frac12 \{I,I\} = 0 .
\eeqn

\subsection{Example: hybrid Chern--Simons theory}\label{s:hcs}

In this subsection we survey a general class of topological-holomorphic theories which have appeared in recent work.
Suppose that $\lie{g}$ is a Lie algebra equipped with a non-degenerate, symmetric and invariant bilinear form $\<-,-\>$.
Suppose that $d' + d$ is a positive odd integer.

The space of fields of hybrid Chern--Simons theory is
\beqn
\alpha \in \Omega^\bu(\R^{d'}) \hotimes \Omega^{0,\bu}(\C^d) \otimes \lie{g}[1] .
\eeqn
The action functional is
\beqn
S(\alpha) = \int_{\R^{d'} \times \C^d} \d^d z \left( \frac12 \<\alpha, \d \alpha\> + \frac16 \<\alpha, [\alpha, \alpha]\> \right) .
\eeqn
Here, $\d$ denotes the de Rham differential in $d' + 2d$ dimensions, but notice that for type reasons only the component $\d_{deRham} + \dbar$ appears in the expression above.
The fact that $\<-,-\>$ is $\lie{g}$-invariant implies that $\{S,S\} = 0$.

We recognize part of the integrand as the usual Chern--Simons Lagrangian.
The appearance of the holomorphic volume element $\d^d z$ implies that only anti-holomorphic derivatives in $\C^d$ will appear in the kinetic term of the action.

The action functional carries cohomological degree $2(d' + d - 3)$, so that this theory is only $\Z$-graded when $d' + d = 3$.
Otherwise, for general $d,d'$ with $d' + d$ an odd integer this theory is merely $\Z/2$ graded (so that a field $\alpha \in \Omega^i (\R^{d'}) \hotimes \Omega^{0,j}(\C^d)\otimes \lie{g}$ is of parity $(i + j - 1) \mod 2$).

We highlight a few other special cases.

\begin{enumerate}
\item $d' = 3, d = 0$. 
This is ordinary (perturbative) topological Chern--Simons theory whose equations of motion describe deformations of the trivial flat $G$-bundle.
\item $d' = 2, d = 1$. 
This is four-dimensional Chern--Simons theory as studied by Costello starting in \cite{Yangian} and further in \cite{CWY1,CWY2}.
Line operators in this theory are controlled by the Yangian quantum group.
\item $d' = 1, d = 2$.
When $\lie{g}=\lie{gl}_k$, there is a deformation of this theory which turns $\C^2$ non-commutative. 
The resulting `non-commutative' Chern--Simons theory was studied by Costello in \cite{CostelloOmega} where it is shown that this theory arises from placing $M$-theory in the (twisted) $\Omega$-background.
Chiral surface operators in this theory encode the $W_{k+\infty}$-vertex algebra.
\item $d'=0, d = 5$.
This is a purely holomorphic theory and has been shown in multiple sources to be the holomorphic twist of ten-dimensional supersymmetric Yang--Mills theory for the Lie algebra $\lie{g}$ \cite{Baulieu, ESW}. 
\item $d'=1, d=4$.
This is the twist of nine-dimensional supersymmetric Yang--Mills theory \cite{ESW}.
\item $d'=2,d=3$.
This is the rank $(1,1)$ twist of eight-dimensional supersymmetric Yang--Mills theory \cite{ESW}.
\item $d'=3,d=2$. 
This is the twist of the rank $2$ twist of seven-dimensional supersymmetric Yang--Mills theory \cite{ESW}.
\item 
$d'=4,d=1$. 
This is the rank $(2,2)$ twist of six-dimensional $\cN=(1,1)$ supersymmetric Yang--Mills theory \cite{ESW}.
\item $d'=5, d=0$.
This purely topological theory is the rank $4$ twist of five-dimensional $\cN=2$ supersymmetric Yang--Mills theory \cite{ESW}.
\end{enumerate}

\subsection{Homotopical BV quantization}
\label{sub:BVquantum}

In this section we swiftly recall the rigorous definition of a perturbative quantum field theory following Costello \cite{CostelloBook} and Costello--Gwilliam \cite{CG2}.
For another great exposition we refer to \cite{LiVertex}.
We refer to these original references for more details.

We start with a classical BV theory on $\R^{d'} \times \C^d$, which for us is always of the form \eqref{eqn:action}, but most of the discussion of this particular section is more general.
Denote the space of fields by 
\beqn\label{space of fields}
\cE \define \Omega^{\bu}(\R^{d'}) \hotimes \Omega^{0,\bu}(\C^d) \otimes V ,
\eeqn
and let $Q = \d_{deRham} + \dbar$ for simplicity of notation. We use $\cE^\vee$ to denote the dual space of $\cE$ as topological vector space.
The classical master equation is then
\beqn \label{CME}
Q I + \frac12 \{I,I\} = 0 .
\eeqn

We remark that the BV bracket $\{-,-\}$ is not defined on all functionals of fields
\beqn
\cO(\cE) = \prod_{n \geq 0} \Sym^n(\cE^\vee) 
\eeqn
as it involves contraction with the distributional form 
\beqn
K_0 \define \omega^{-1} \in \Sym^2(\br \cE)[-1] 
\eeqn
with $\omega$ as in equation \eqref{eqn:omega}. 
The bar on the right hand side means we consider $K_0$ as a distributional section.
The bracket is nevertheless well-defined on \textit{local} functionals
\beqn
\cO_{loc}(\cE)
\eeqn
which, by definition, are equivalence classes of Lagrangian densities where two densities are equivalent if they agree up to a total derivative \cite{CG2}. For simplicity of exposition, we assume the interaction has the following form: 
\begin{equation}\label{weaker assumption}
I(\phi)=\sum_{k=3}^{\infty}\int_{\R^{d'} \times \C^d} \d^d z \,f_{k} \prod_{i=1}^{k}D_{i}\phi\in \cO_{loc}(\cE),
\end{equation}
where $D_{i}$ is a differential operator with constant coeffients, and it only contains derivatives with respect to the holomorphic coordinates $\{z_{i}\}_{i=1}^{d}$. Here, $f_{k}$ are smooth differential forms with compact supports. 

\begin{rmk}
    Our main theorems Theorem \ref{main theorem1}, \ref{main theorem2} and \ref{main theorem3} still hold if we replace the assumption $(\ref{weaker assumption})$ by $(\ref{stronger assumption})$. In particular, the support of $I$ can be noncompact. See the fake heat kernel method developed in \cite{CostelloBook}.
\end{rmk}
The most important operator in the Batalin--Vilkovisky formalism it the BV Laplacian $\triangle$.
This is naively defined as the contraction with the distributional kernel $K_0$, but just as with the BV bracket this is not well-defined on all functionals of the fields.
The main idea involved with homotopical renormalization is to consider the new kernel
\beqn
K_\Phi = K_0 - Q P_\Phi
\eeqn
where 
\beqn
P_\Phi \define \frac12 Q^{GF} \Phi 
\eeqn
is the \textit{propagator} associated to a \textit{parametrix} $\Phi$ \cite[\S7.2]{CG2}.
Here $Q^{GF}$ is a `gauge fixing' operator for the theory, which for us is $Q^{GF} = \d_{deRham}^\ast + \dbar^\ast$.
The key property of $K_\Phi$ is that it is a smooth (non-distributional) section of $\Sym^2(\cE)[-1]$ so that the operator $\triangle_\Phi$ defined by contraction with $K_\Phi$ is well-defined.

We will not recall the full definition of a parametrix as in \cite{CG2}, but we exhibit how they are used to define a quantum field theory.
Below, we show how to produce parametrices using the heat equation.
Following \cite[definition 7.2.9.1]{CG2}, a \defterm{quantum field theory} is a family of functionals 
\beqn
\{I[\Phi]\} \subset \cO(\cE)[[\hbar]],
\eeqn
one for each parametrix $\Phi$, which satisfies a number of axioms. 
The two most important axioms are:
\begin{enumerate}
\item \textit{Renormalization group flow}.
For every two parametrices $\Phi, \Psi$ the family must satisfy
\beqn
I[\Phi] = W(P(\Phi) - P(\Psi), I[\Psi])
\eeqn
where $W(-,-)$ is the Feynman diagram expansion (See Definition \ref{Feynman expansion}).
\item \textit{Quantum master equation}.
For each $\Phi$ one has
\beqn
(Q + \hbar \triangle_\Phi) e^{I[\Phi]/\hbar} = 0 .
\eeqn  
Obstructions to this equation holding are called \defterm{anomalies}.
\item \textit{Classical limit}. 
If $I$ is the interaction describing the classical field theory, then we say that $\{I[\Phi]\}$ is a quantization of $I$ if $I = \lim_{\Phi \to 0} I[\Phi] \mod \hbar$.
\end{enumerate}

We now point out how parametrices can be produced from the heat equation.
Let
\beqn
\Delta = [Q,Q^{GF}] = \bar{\partial} \circ \bar{\partial}^{\ast} + \bar{\partial}^{\ast}
   \circ \bar{\partial} + \d_{deRham} \circ \d_{deRham}^{\ast} +\d_{deRham}^{\ast} \circ \d_{deRham} 
\eeqn
be the standard Laplacian on flat space $\R^{d'} \times \C^d$.
Here, we recall that
\[ \left\{\begin{array}{l}
     \bar{\partial}^{\ast} = - 2 \sum_{i = 1}^d \frac{\partial}{\partial z_i}
     \iota_{\frac{\partial}{\partial \bar{z}_i}} \\
     \d_{deRham}^{\ast} = - \sum_{i = 1}^{d'} \frac{\partial}{\partial x_i}
     \iota_{\frac{\partial}{\partial x_i}}
   \end{array}\right. \]
are the formal adjoints to the de Rham and Dolbeault operators with respect to the flat metric.
This operator is not to be confused with the BV Laplacian $\triangle$ acting on functionals of fields.

We consider the associated heat kernel
\beqn
H (t, z, x) = \frac{1}{2^{d + d'}
   (\pi t)^{d + \frac{d'}{2}}} e^{- \frac{\overset{d}{2 \tmmathbf{\underset{i
   = 1}{\sum}}} z_i \overline{z_i} + \overset{d'}{\tmmathbf{\underset{i =
   1}{\sum}}} x_i^2}{4 t}} \d^d \bar{z} \d^{d'} x, \text{\quad$t > 0$} 
\eeqn
This is a differential form solving the heat equation with initial condition
\beqn
\left\{\begin{array}{l}
     \left( \frac{\partial}{\partial t} + \Delta \right) H (t, z,x) = 0\\
      H (t, z,x) \xto{t\to 0} \delta (z, x) \d^d \bar{z} \d^{d'} x,
   \end{array}\right.
\eeqn
with $\delta(z,x)$ the $\delta$-distribution at $0 \in \C^d \times \R^{d'}$. 
Define
\begin{align}\label{heat kernel}
    &K_t (z-w, x-y)\\ \notag\define & H(t,z-w,x-y) \otimes \lie{c}_V \in \cE(\R^{d'} \times \C^d) \hotimes \cE(\R^{d'} \times \C^d) \hotimes C^\infty(\R_{+})\otimes V\otimes V,
\end{align}
where $\lie{c}_V$ is tensor dual to the pairing on $V$.\footnote{Fix a basis $\{v_i\}$ orthonormal for the pairing. Then $\lie{c}_V = \sum_i v_i \otimes v_i$.}
The parametrix associated to this heat kernel is a distribution
\beqn
\Phi_L \define \int_{0}^L \d t \, K_t,
\eeqn
and the propagator is
\beqn
P_{0<L} = P(\Phi_L) =  \int_{t=0}^L (\bar{\partial}^{\ast} \otimes \id + \d_{deRham}^{\ast} \otimes \id) K_t \d t .
\eeqn
Also, let $\triangle_L$ denote contraction with $K_L$ and $P_{\epsilon<L} = P(\Phi_L) - P(\Phi_\epsilon)$. We use $\{-,-\}_{L}$ to denote the BV bracket defined by $K_L$.

In what follows, we will produce quantizations from Feynman diagram expansions using propagators built explicitly from heat kernels.
This means that we will produce a family of functionals $\{I[L]\}$ for each `scale' $L > 0$ which satisfies the following key properties
\begin{enumerate}
\item \textit{Heat kernel renormalization group flow}.
For every $\epsilon,L > 0$ the family must satisfy
\beqn
I[L] = W(P_{\epsilon<L}, I[\epsilon])
\eeqn
where $W(-,-)$ is the Feynman diagram expansion as above.
\item \textit{Heat kernel quantum master equation}.
For each $L>0$ one has
\beqn
(Q + \hbar \triangle_L) e^{I[L]/\hbar} = 0 .
\eeqn 
\item \textit{Heat kernel classical limit}. 
If $I$ is the interaction describing the classical field theory, then we say that $\{I[L]\}$ is a quantization of $I$ if $I = \lim_{L \to 0} I[L] \mod \hbar$.
\end{enumerate}
For the additional axioms this family must satisfy we refer to \cite[definition 7.1.2]{CostelloBook}, \cite[definition 7.2.9.1]{CG2}.
We refer to such a family as a \defterm{heat kernel quantum field theory}.
Every heat kernel defines a parametrix $\Phi_L$.
Given a heat kernel quantum field theory, we obtain a quantum field theory in the sense above by the formula
\beqn
I[\Phi] = W(P(\Phi) - P(\Phi_L), I[L]) .
\eeqn

Based on the terminology we introduced, our main theorem can be stated as follows:
\begin{thm}\label{main theorem3}
    Given the space of fields $\cE$ defined in $(\ref{space of fields})$, and an interaction $I$, which satisfies the assumption $(\ref{weaker assumption})$ and the classical master equation $(\ref{CME})$. If $d'>1$, then there exists a heat kernel quantum field theory $\{I[L]\}$, such that \[\lim_{L\rightarrow 0}I[L]=I.\]
\end{thm}
From above theorem, one can construct a factorization algebra corresponding to the quantum observables of the quantum field theory. See \cite{CG2} for details. 

\subsection{Combinatorics of stable graphs}

To separate the analytic part of perturbation theories, we recall some basic
combinatorics of graphs. For more detailed explanations, please refer \cite{CostelloBook}.

\begin{dfn}
  A stable graph is a non-directed graph $\gamma$, possibly with external
  edges; and for each vertex $v$ of $\gamma$ an element $g (v) \in
  \mathbb{Z}_{\geqslant 0}$, called the genus of the vertex $v$; with the
  property that every vertex of genus $0$ is at least trivalent, and every
  vertex of genus $1$ is at least 1-valent.
  
  If $\gamma$ is a stable graph, the genus $g (\gamma)$ of $\gamma$ is
  defined by
  \[ g (\gamma) = b_1 (\gamma) + \sum_{v \in V (\gamma)} g (v) \]
  where $b_1 (\gamma)$ is the first Betti number of $\gamma$, $V (\gamma)$ is
  the set of all vertices of $\gamma$. An ordering of $\gamma$ is an ordering
  of the external edges, internal edges and vertices of $\gamma$. We use
  $\tmop{Ord} (\gamma)$ to denote the set of all ordering of $\gamma$.
\end{dfn}

\begin{dfn}
  An automorphism of a stable graph $\gamma$ is an automorphism of graph which
  preserve the genus of vertices of $\gamma$. We denote the set of
  automorphisms by $\tmop{Aut} (\gamma)$.
\end{dfn}

Let
\[ \mathcal{O}^+ (\mathcal{E}) [[\hbar]] \subset \mathcal{O} (\mathcal{E})
   [[\hbar]] \]
be the subspace of those functionals which are at least cubic modulo $\hbar$.

If $I \in \mathcal{O} (\mathcal{E}) [[\hbar]]$, we have
\[ I = \sum_{i, k \geqslant 0} \hbar^i I_{i, k}, \]
where $I_{i, k}$ is homogeneous of degree $k$ as a polynomial on
$\mathcal{E}$.

If $f\in \mathcal{O}(\mathcal{E})$ is homogeneous of degree $k$, then it defines an $S_k$-invariant continuous linear map
\[
\mathrm{D}^k f:\mathcal{E}^{\hat{\otimes} k}\to\C,\qquad
\phi_1\otimes\cdots\otimes\phi_k\mapsto k!\,f(\phi_1,\dots,\phi_k).
\]
Thus, if we expand $I \in \mathcal{O} (\mathcal{E}) [[\hbar]]$ as above, we have a collection of linear maps $\mathrm{D}^kI_{i,k}$.

Let $\gamma$ be a stable graph, with $k$ external edges. We choose an
ordering of $\gamma$. We will define
\[ w_{\gamma} (P_{\epsilon < L}, I) \in \tmop{Sym}^k (\mathcal{E}^{\vee})
   \subset \tmop{Hom} (\mathcal{E}^{\hat{\otimes} k}, \C) . \]
The rule is as follows. Let $H (\gamma)$, $T (\gamma)$, $E (\gamma)$, and $V
(\gamma)$ refer to the sets of half-edges, external edges, internal edges,
vertices of $\gamma$, respectively. Putting a propagator $P_{\epsilon < L}$ at
each internal edge of $\gamma$, putting $a_i$ at $i$th external edge of
$\gamma$, we can obtain an element of
\[ \mathcal{E}^{\hat{\otimes} | E (\gamma) |} \otimes
   \mathcal{E}^{\hat{\otimes} | E (\gamma) |} \otimes
   \mathcal{E}^{\hat{\otimes} | T (\gamma) |} \cong \mathcal{E}^{\hat{\otimes}
   | H (\gamma) |} . \]
Putting $\mathrm{D}^{k'}I_{i,k'}$ at each vertex of valency $k'$ and genus $i$ gives us an
element of
\[ \tmop{Hom} (\mathcal{E}^{\hat{\otimes} | H (\gamma) |}, \C) . \]
Contracting these two elements yields a number $\tilde{w}_{\gamma}
(P_{\epsilon < L}, I) (a_1, \ldots, a_k)$. This define an element
\[ \tilde{w}_{\gamma} (P_{\epsilon < L}, I) \in \tmop{Hom}
   (\mathcal{E}^{\hat{\otimes} k}, \C) . \]
We define
\[ w_{\gamma} (P_{\epsilon < L}, I)=\frac{1}{k!}\sum_{\sigma\in S_k}\sigma(\tilde{w}_{\gamma} (P_{\epsilon < L}, I)) \in \tmop{Sym}^k (\mathcal{E}^{\vee}) \]
as the symmetrization of $\tilde{w}_{\gamma} (P_{\epsilon < L}, I)$.

\begin{dfn}\label{Feynman expansion}
  The Feynman graph expansion $W (-, -)$ is defined by the following formula:
  \[ W (P_{\epsilon < L}, I) = \sum_{\gamma} \frac{1}{| \tmop{Aut} (\gamma) |}
     \hbar^{g (\gamma)} w_{\gamma} (P_{\epsilon < L}, I) \in \mathcal{O}^+
     (\mathcal{E}) [[\hbar]] . \]
\end{dfn}

To write down the quantum master equation in terms of Feynman graph expansion,
we introduce the following technical definitions:

\begin{dfn}
  Let $\gamma$ be a stable graph, $e \in E (\gamma)$, $v \in V (\gamma)$, $P'
  \in \tmop{Sym}^2 (\mathcal{E})$, and $I' \in \mathcal{O}^+ (\mathcal{E})[[\hbar]]$.
  \begin{enumeratenumeric}
    \item We define $w_{\gamma, e} (P_{\epsilon < L}, P', I)$ in the same way
    as $w_{\gamma} (P_{\epsilon < L}, I)$, except the distinguished edge $e$
    is labeled by $P'$, whereas all other edges are labelled by $P$.
    
    \item We define $w_{\gamma, v} (P_{\epsilon < L}, I, I')$ in the same way
    as $w_{\gamma} (P_{\epsilon < L}, I)$, except the distinguished vertex $v$
    is labeled by $I'$, whereas all other vertices are labelled by $P$.
    
    \item We use $\tmop{Aut} (\gamma, e)$ to denote the set of automorphisms
    which preserve $e$. Likewise, we use $\tmop{Aut} (\gamma, v)$ to denote
    the set of automorphisms which preserve $v$.
  \end{enumeratenumeric}
\end{dfn}

We have the following lemma:

\begin{lem}
  \label{weight properties}Let $\gamma$ be a stable graph, we define
  $w_{\gamma, e} (P_{\epsilon < L}, Q (P_{\epsilon < L}), I)$ The following
  are true:
  \begin{enumeratenumeric}
    \item
    \begin{eqnarray*}
      &  & Q W (P_{\epsilon < L}, I)\\
      & = & \sum_{\gamma, e} \frac{1}{| \tmop{Aut} (\gamma, e) |} \hbar^{g
      (\gamma)} w_{\gamma, e} (P_{\epsilon < L}, Q (P_{\epsilon < L}), I)\\
      & + & \sum_{\gamma, v} \frac{1}{| \tmop{Aut} (\gamma, v) |} \hbar^{g
      (\gamma)} w_{\gamma, v} (P_{\epsilon < L}, I, Q I) .
    \end{eqnarray*}
    \item
    \begin{eqnarray*}
      &  & \frac{1}{2} \{ W (P_{\epsilon < L}, I), W (P_{\epsilon < L}, I)
      \}_L + \hbar \triangle_L W (P_{\epsilon < L}, I)\\
      & = & \sum_{\gamma, e} \frac{1}{| \tmop{Aut} (\gamma, e) |} \hbar^{g
      (\gamma)} w_{\gamma, e} (P_{\epsilon < L}, K_L, I) .
    \end{eqnarray*}
    \item
    \begin{eqnarray*}
      &  & \sum_{\gamma, v} \frac{1}{| \tmop{Aut} (\gamma, v) |} \hbar^{g
      (\gamma)} w_{\gamma, v} \left( P_{\epsilon < L}, I, \frac{1}{2} \{ I, I
      \} \right)\\
      & = & \lim_{t \rightarrow 0} \sum_{\gamma, e} \frac{1}{| \tmop{Aut}
      (\gamma, e) |} \hbar^{g (\gamma)} w_{\gamma, e} (P_{\epsilon < L}, K_t,
      I).
    \end{eqnarray*}
  \end{enumeratenumeric}
\end{lem}

\begin{proof}
  The first two identities can be proven by graph combinatorics. As an
  example, We prove the first identity. For a stable graph $\gamma$, we have
  \[ Q w_{\gamma} (P_{\epsilon < L}, I) = \sum_{e \in E (\gamma)} w_{\gamma,
     e} (P_{\epsilon < L}, Q (P_{\epsilon < L}), I) + \sum_{v \in V (\gamma)}
     w_{\gamma, v} (P_{\epsilon < L}, I, Q I) \]
  There is an action of the automorphism group $\tmop{Aut} (\gamma)$ on $E
  (\gamma)$. The quotient set is the isomorphism class of stable graphs with a
  special edge such that the underling graph is $\gamma$. We use $E' (\gamma)$
  to denote this set. Likewise, we use $V' (\gamma)$ to denote quotient set of
  $\tmop{Aut} (\gamma)$ on $V (\gamma)$. Then we have
  \begin{eqnarray*}
    &  & Q w_{\gamma} (P_{\epsilon < L}, I)\\
    & = & \sum_{e \in E' (\gamma)} \frac{| \tmop{Aut} (\gamma) |}{|
    \tmop{Aut} (\gamma, e) |} w_{\gamma, e} (P_{\epsilon < L}, Q (P_{\epsilon
    < L}), I) + \sum_{v \in V' (\gamma)} \frac{| \tmop{Aut} (\gamma) |}{|
    \tmop{Aut} (\gamma, v) |} w_{\gamma, v} (P_{\epsilon < L}, I, Q I) .
  \end{eqnarray*}
  We multiply both sides of above identity by $\frac{1}{| \tmop{Aut} (\gamma)
  |} \hbar^{g (\gamma)}$, and sum them over all stable graphs. The final
  result is the first identity we want to prove.

  To prove the
  last identity, we notice that the pull back of $K_t (z-w, x-y)$ defined by \ref{heat kernel} to the diagonal $\{(z,x,w,y)\in\R^{d'} \times \C^d\times\R^{d'} \times \C^d:z=w, x=y\}$ is zero, so
  \[ \lim_{t \rightarrow 0} \triangle_t I = 0. \]
  Therefore, the summation of right hand side is nonzero only if $e$ is not a
  self-loop of $\gamma$. Then we notice that there is a bijection $f$ between
  graphs with one special edge and graphs with one special vertice:
  \[ f : (\gamma, e) \rightarrow (\gamma / e, v_e), \]
  where $\gamma / e$ is the graph obtained by contract $e$ to a single vertex
  $v_e$. The last identity follows.
\end{proof}

\begin{cor}\label{weight QME}
  Assume we have the classical master equation
  \[ Q I + \frac{1}{2} \{ I, I \} = 0. \]
  The quantum master equation is satisfied if, for any stable graph $\gamma$,
  the following equality holds:
  \begin{equation}\label{weight QME2}
  \lim_{\epsilon\rightarrow 0}\sum_{e \in E (\gamma)} \left( w_{\gamma, e} (P_{\epsilon < L}, Q (P_{\epsilon < L}),
     I) - \lim_{t \rightarrow 0}  w_{\gamma, e} (P_{\epsilon <
     L}, K_t, I) + w_{\gamma, e} (P_{\epsilon < L}, K_L, I) \right) = 0. 
\end{equation}
\end{cor}
\begin{proof}
  The existence of the limit when $\epsilon\rightarrow 0$
  is nontrivial. This follows from our ultraviolet finiteness theorem in later
  sections. Now, we will assume this fact.
  
  If $\gamma$ is a stable graph, there is an action of the automorphism group
  $\tmop{Aut} (\gamma)$ on $E (\gamma)$. The quotient set is the isomorphism
  class of stable graphs with a special edge such that the underling graph is
  $\gamma$. We use $E' (\gamma)$ to denote this set. Then we have
  \begin{eqnarray*}
    0 & = & \lim_{\epsilon\rightarrow 0}\sum_{e \in E (\gamma)} \left( w_{\gamma, e} (P_{\epsilon < L}, Q (P_{\epsilon <
    L}), I) - \lim_{t \rightarrow 0}  w_{\gamma, e}
    (P_{\epsilon < L}, K_t, I) + w_{\gamma, e} (P_{\epsilon < L}, K_L, I) \right)\\
    & = & \lim_{\epsilon\rightarrow 0}\sum_{e \in E' (\gamma)} \frac{| \tmop{Aut} (\gamma) |}{|
    \tmop{Aut} (\gamma, e) |} \left( w_{\gamma, e} (P_{\epsilon < L}, Q (P_{\epsilon < L}),
    I) - \lim_{t \rightarrow 0}  w_{\gamma, e} (P_{\epsilon <
    L}, K_t, I) \right.\\ 
    & + &\nobracket w_{\gamma, e} (P_{\epsilon < L}, K_L, I)\Bigr).
  \end{eqnarray*}
  If we multiply both sides of above identity by $\frac{1}{| \tmop{Aut}
  (\gamma) |} \hbar^{g (\gamma)}$, and sum them over all graphs, we have
  \begin{eqnarray*}
    0 & = & \lim_{\epsilon\rightarrow 0}\sum_{\gamma, e} \frac{1}{| \tmop{Aut} (\gamma, e) |} \left(
    w_{\gamma, e} (P_{\epsilon < L}, Q (P_{\epsilon < L}), I) - \lim_{t \rightarrow 0} w_{\gamma, e} (P_{\epsilon < L}, K_t, I) \right.\\
    & + & \nobracket w_{\gamma, e} (P_{\epsilon < L}, K_L, I)\Bigr).
  \end{eqnarray*}
  The quantum master equation follows from Lemma \ref{weight properties}:
  \[ \lim_{\epsilon\rightarrow 0}\left( Q W (P_{\epsilon < L}, I) + \frac{1}{2} \{ W (P_{\epsilon < L}, I), W (P_{\epsilon < L}, I)
     \}_L + \hbar \triangle_L W (P_{\epsilon < L}, I)\right) = 0. \]
\end{proof}
\begin{rmk}
    The corollary above shows that one can verify the quantum master equation by checking each graph separately.
\end{rmk}
For the convenience of our proof of main theorems, we notice the following
fact:
\[ K_t (z - w, x - y) = H (t, z - w, x - y) \otimes \mathfrak{c}_V \]
is a simple tensor in
\[ \tmop{Sym}^2 (\mathcal{E}) \cong \tmop{Sym}^2 (\Omega^{\bullet} (\C^d \times \R^{d'})) \otimes \tmop{Sym}^2 (V).
\]
Moreover, by our assumption $(\ref{stronger assumption})$ on $I=\sum_{k=3}^{\infty}I_{k}$, each $I_{k}$ is a simple tensor, so we have the following decomposition of $w_{\gamma} (P_{\epsilon < L}, I)$:
\begin{equation}\label{factorize Feynman graph}
 w_{\gamma} (P_{\epsilon < L}, I) = w_{\gamma}^{\tmop{an}} (P_{\epsilon <
   L}, I) \otimes w_{\gamma}^{\tmop{al}} (P_{\epsilon < L}, I) .
\end{equation}
We call $w_{\gamma}^{\tmop{an}} (P_{\epsilon < L}, I)$ the analytic part of
Feynman graph integral, $w_{\gamma}^{\tmop{al}} (P_{\epsilon < L}, I)$ the
algrbraic part of Feynman graph integral.

\begin{rmk}
    If we replace the assumption $(\ref{weaker assumption})$ by $(\ref{stronger assumption})$, we don't have the decomposition $(\ref{factorize Feynman graph})$. In this case, our main theorems Theorem \ref{main theorem1}, \ref{main theorem2} and \ref{main theorem3} still hold. In this case, the Feynman graph integrals are finite sums of terms similar to $(\ref{factorize Feynman graph})$, and we can apply our techniques in later sections to each term respectively.
\end{rmk}

In the following sections, we will only care about the analytic part of
Feynman graph integrals.

\begin{rmk}
  Note $I$ is a local functional which depends only on holomorphic
  derivatives. 
  The integrand of $w_{\gamma}^{\tmop{an}} (P_{\epsilon < L}, I) (a_1, \ldots
a_k)$ will be a product of holomorphic derivatives of $P_{\epsilon < L}$ and a
smooth differential form with compact support. This motivate the definition of
Feynman graph integrals in next section.
\end{rmk}

\section{Feynman graph integals on $\R^{d'} \times \C^d$}
\label{sec:feynman}

In this section we study the the analytic part of Feynman graph
integrals on $\R^{d'} \times \C^d$.
First, we will prove that topological-holomorphic theories are `UV' finite, in the sense defined below.
Then, we obtain some vanishing results for integrals over boundaries of
compactified Schwinger spaces which will allow us to prove that topological-holomorphic theories on $\R^{d'} \times \C^d$ admit a quantization to all orders in perturbation theory provided $d' > 1$.
We refer to appendix \ref{Schwinger spaces} for a
introduction to Schwinger spaces and the sort of compactifications that we make use of in this section.
Throughout this section, we will omit the bundle factors (or algebraic factors) that appear in Feynman graph integrals without further mention.

We will use coordinates
$z=(z_i)_{1 \leqslant i \leqslant d}$ and $x = (x_j)_{1 \leqslant j
\leqslant d'}$ for holomorphic and smooth coordinates on $\C^d$ and
$\R^{d'}$, respectively. 
The notation 
\beqn
\d^d \bar{z} \, \d^{d'} x = \left( \underset{i = 1}{\overset{d}{\prod}} \d \bar{z}_i
\right) \left( \underset{j = 1}{\overset{d'}{\prod}} \d x_j \right) .
\eeqn 
refers to the top anti-holomorphic form on $\C^d \times
\R^{d'}$.
The Dolbeault differential on $\C^d$ is denoted by $\dbar$ and the de Rham differential on $\R^{d'}$ is denoted simply by $\d_{deRham}$.

Note that any (distributional) differential form on $\R^{d'} \times \C^d$ can be written as 
\beqn
\alpha = \sum f_{I \br J K} (z,\zbar,x) \d^I z \d^{\br J}\zbar \d^K x
\eeqn
By an \defterm{anti-holomorphic} form on $\R^{d'} \times \C^d$ we mean such a form with index $I = 0$.
We will refer to the \defterm{Hodge type} of a (distributional) anti-holomorphic differential form as $(|\br J|, |K|)$.

\subsection{Propagators and Feynman graph integrals}
\label{sub:props}

The basic ingredient in the definition of Feynman graph integrals is the propagator, which is a particular distributional
valued differential form which is a Green's function for the operator $\dbar + \d_{deRham}$.

For $p \in \R^{d'} \times \C^d$, recall that the $\delta$-function at $p$ is the form degree $2d+d'$ distribution $\delta_p$ defined by $\delta_p (f) = f(p)$ for all $f \in C^\infty_c(\R^{d'} \times \C^d)$.
The restriction of $\delta_0$ along the difference map
\beqn
(\R^{d'} \times \C^d) \times (\R^{d'} \times \C^d) \to \R^{d'} \times \C^d ,\quad (z,x;w,y) \mapsto (z-w,x-y) 
\eeqn
is denoted $\delta(z-w, x-y) \d^{d} (z-w)$.\footnote{Sometimes to emphasize the anti-holomorphic dependence we will write $\delta(z-w, \zbar - \wbar, x-y)$.}
Note that we have extracted the holomorphic $d$-form component so that $\delta(z-w, x-y)$ has form degree $d+d'$.

\begin{dfn}
  An {\defterm{ordinary propagator}} on $\R^{d'} \times \C^d$
  is a distributional valued differential form
  \[ \til{P} (z - w, \zbar -
     \wbar, x - y) \in
     \mathcal{D}^{\bu} (\R^{d'} \times \C^d \times
     \R^{d'} \times \C^d), \]
  such that the following equation holds:
  \[ (\dbar + \d_{deRham}) \, \til{P} (z - w, \zbar -
     \wbar, x - y)  = \delta(z-w, \zbar - \wbar, x-y) .
  \]
\end{dfn}

\begin{rmk}
The form degree of $\til P$ is $d+d'-1$.
By construction, $\til P$ is an anti-holomorphic distributional form.
In the model we will use $\til P$ is a sum of forms of Hodge type $(d,d'-1)$ and $(d-1,d')$.
\end{rmk}

The choice of propagator is not unique, but we have the following standard choice.

\begin{dfn}
  The {\defterm{generalized Bochner-Martinelli kernel}} is the distributional form
  \begin{eqnarray*}
    &  & \frac{2^d \Gamma \left( d + \frac{d'}{2} \right)}{\pi^{d +
    \frac{d'}{2}}} \cdot \frac{1}{(2 | z - w |^2 + | x -
    y |^2)^{d + \frac{d'}{2}}} \cdot\\
    &  & \left( 2\sum_{i = 1}^d (- 1)^{i - 1} (\overline{z_i - w_i}) \left(
    \prod_{j \neq i} d (\overline{z_j - w_j}) \right) \d^{d'} (x - y) +
    \right.\\
    &  & \left. \sum_{i = 1}^{d'} (- 1)^{d + i - 1} (x_i - y_i) \d^d
    (\overline{z - w}) \left( \prod_{j \neq i} \d (x_j - y_j) \right) \right)
  \end{eqnarray*}
\end{dfn}

Standard computations give the following.

\begin{lem}
The generalized Bochner--Martinelli kernel is an ordinary propagator on~$\R^{d'} \times \C^d$. 
We will denote it by $\til{P}_{0, + \infty} (z-w,
  \zbar - \wbar, x-y)$.
\end{lem}

This generalized Bochner-Martinelli kernel can be constructed from the solution to the heat equation.
We recall some notations from section \ref{sub:BVquantum}.

\begin{dfn}
  The {\defterm{regularized ordinary propagator}} $\til{P}_{\varepsilon,
  L}$ is defined by the following formula:
\beqn
\til{P}_{\varepsilon, L} (z-w,\zbar-\wbar, x -
    y) = \int_{t=\varepsilon}^L (\bar{\partial}_z^{\ast} + \d_{deRham,x}^{\ast}) H (t,
    z-w,\zbar-\wbar,
    x-y) \, \d t 
\eeqn
defined for $\varepsilon, L > 0$.
Here $\bar{\partial}_z^{\ast}$ and $\d_{deRham,x}^{\ast}$ are differential operators acting on the
  variables $z$ and $x$, respectively.
\end{dfn}

\begin{prop}
  The following equality holds:
\beqn
\til{P}_{0, + \infty} (z-w,
  \zbar - \wbar, x-y) =
     \lim_{\underset{L \rightarrow + \infty}{\varepsilon \rightarrow 0}}
     \til{P}_{\varepsilon, L} (z-w,
  \zbar - \wbar, x-y) .
  \eeqn
\end{prop}

\begin{proof}
 This can be shown by direct computation.
\end{proof}

We can now define the sort of Feynman graph integrals associated with this choice of propagator.
For $v$ an integer we denote by $(z^1,z^2,\ldots,z^v;x^1,\ldots,x^v)$ coordinates on the $v$th product of spacetime $(\R^{d'} \times \C^d)^v = \R^{d' v} \times \C^{d v} $.
In particular, for each $1 \leq i \leq v$ the pair $(z^i,x^i) = (z^i_1,\ldots,z^i_d; x^i_1,\ldots,x^{i}_{d'})$ is a coordinate on $\R^{d'} \times \C^d$. 
We refer to appendix \ref{graph theory} for further conventions with graphs that we use in this section.

\begin{dfn}
The \defterm{configuration space} of a graph $\Gamma$ is
  \begin{eqnarray*}
\tmop{Conf} (\Gamma) & = & \left\{ \left. \left( (z^1, x^1),
    (z^2, x^2), \ldots, \left(
    z^{| \Gamma_0 |}, x^{| \Gamma_0 |}
    \right) \right)  \right|
    \nobracket \right. \left. (z^i, x^i) \neq
    (z^j, x^j) \text{ for  } i \neq j
    \right\} \\ & \subset & (\R^{d'} \times \C^d)^{| \Gamma_0 |}
  \end{eqnarray*}
  We interpret $(z^i, x^i)$ as the coordinate at the $i$th vertex of $\Gamma$.
\end{dfn}

We fix the following data:
\begin{itemize}
\item A decorated graph $(\Gamma, n)$ (as explained in appendix \ref{graph theory}),
\item positive numbers $0 < \varepsilon < L$, and 
\item a smooth, compactly supported, differential form $\Phi \in \Omega^{\bu}_c ((\R^{d'} \times \C^d)^{| \Gamma_0 |})$.
\end{itemize}
The main object of study is the {\defterm{Feynman graph integral}}
\beqn
  W_{0}^{L} ((\Gamma, n), \Phi)
\eeqn
formally defined by
\begin{multline}
(- 1)^{^{\frac{d + d' - 1}{2} | \Gamma_1 | (| \Gamma_1 | - 1) + |
    \Gamma_1 |}}\int_{\op{Conf}(\Gamma)}
    \prod_{e \in \Gamma_1 } \partial_{z^{h (e)}}^{n (e)}
    \til{P}_{0, L} (z^{h (e)} -
    z^{t (e)}, \bar{z}^{h (e)} -
    \bar{z}^{t (e)}, x^{h (e)} -
    x^{t (e)}) \wedge \Phi .
\end{multline}
This integral may not exist, but a regularized version always does.

\begin{dfn}
Let $(\Gamma,n)$ and $\Phi$ be as above. 
We define the {\defterm{regularized Feynman graph integral}} 
\beqn
  W_{\varepsilon}^L ((\Gamma, n), \Phi)
\eeqn
on $\R^{d'} \times \C^d$ to be the following integral:
\begin{multline}
(- 1)^{^{\frac{d + d' - 1}{2} | \Gamma_1 | (| \Gamma_1 | - 1) + |
    \Gamma_1 |}}\int_{(\R^{d'} \times \C^d)^{| \Gamma_0 |}}
    \prod_{e \in \Gamma_1 } \partial_{z^{h (e)}}^{n (e)}
    \til{P}_{\varepsilon, L} (z^{h (e)} -
    z^{t (e)}, \bar{z}^{h (e)} -
    \bar{z}^{t (e)}, x^{h (e)} -
    x^{t (e)}) \wedge \Phi .
\end{multline}
Here $\partial_{z^{h (e)}}^{n (e)} = \partial_{z^{h (e)}_1}^{n_{1, e}}
  \partial_{z^{h (e)}_2}^{n_{2, e}} \ldots \partial_{z^{h (e)}_i}^{n_{i, e}}
  \ldots \partial_{z^{h (e)}_d}^{n_{d, e}}$ is a holomorphic differential operator with
  constant coefficients which only involves coordinates at the vertex~$h (e)$.
\end{dfn}

The main result of this section is to show that the following.

\begin{thm}\label{thm:renormalization}
The limit
\beqn
W_0^L((\Gamma,n),\Phi) = \lim_{\varepsilon \to 0} W_{\varepsilon}^L ((\Gamma,
n), \Phi) ,
\eeqn
exists.
\end{thm}

In other words, we will prove the ultraviolet finiteness of such Feynman graph integrals on $\R^{d'} \times \C^d$.
Once we know that this limit exists it follows that the original Feynman graph integral satisfies $W_0^{+\infty}((\Gamma,n),\Phi) < + \infty$.
We will prove the existence of this limit by using compactification of Schwinger spaces (see appendix
\ref{Schwinger spaces}).

To do this, we recast $W_0^L ((\Gamma, n), \Phi)$ in terms of the following
propagator in Schwinger spaces:

\begin{dfn}
  Given $t > 0$, The {\defterm{propagator in Schwinger space}} $P_t$ is
  defined by the following formula:
\begin{eqnarray}\label{schwinger propagator}
   & &P_t (z - w, \zbar -
    \wbar, x-y)\nonumber\\
    &=& - \d t \wedge (\bar{\partial}_z^{\ast} + \d_{deRham,x}^{\ast}) H(t, z - w, \zbar -
    \wbar, x-y) + H(t, z - w, \zbar -
    \wbar, x-y) ,
\end{eqnarray}
or, compactly $P_t = - \d t (\bar{\partial}_z^{\ast} + \d_{deRham,x}^{\ast}) H + H$.
  We will simply call $P_t$ the propagator if there is no ambiguity.
\end{dfn}

One important reason to introduce this propagator is the following lemma.

\begin{lem}
  \label{lem:u}Let $u = \frac{\zbar-\wbar}{2 t}$, $v = \frac{x-y}{2 \sqrt{t}}$.
  Then
  \[ P_t (z - w, \zbar -
    \wbar, x-y) =
     \frac{1}{\pi^{d + \frac{d'}{2}}} e^{- (z - w) \cdot u -
    v \cdot v} \d^d u \d^{d'} v \]
where $(z-w) \cdot u = \sum_{i=1}^d (z_i - w_i)u_i$ and $v \cdot v = \sum_{j=1}^{d'} v_jv_j$.
\end{lem}

\begin{proof}
  This can be shown by direct computation.
\end{proof}

We have two additional useful properties for propagator in Schwinger space:

\begin{lem}
  \label{closeness of propagator}{\tmdummy}
  
  \begin{enumeratenumeric}
    \item Let $\d_t$ be the de Rham differential on Schwinger space, then
    \beqn (\d_t + \bar{\partial} + \d_{deRham}) P_t (z - w, \zbar -
    \wbar, x-y) = 0. 
    \eeqn
    \item Define the vector fields $Eu_t = t \frac{\partial}{\partial t}$, $Eu_{\zbar} = \sum_{i=1}^d \zbar_i \frac{\partial}{\partial \zbar}$, $Eu_{x} = \sum_{j=1}^{d'} x_j \frac{\partial}{\partial x_j}$.
    Then one has the equality 
\beqn
\left( \iota_{Eu_t} + \iota_{Eu_{\zbar}} + \iota_{Eu_{\wbar}} + \iota_{Eu_x} + \iota_{Eu_y} \right)P_t (z - w, \zbar -
    \wbar, x-y) = 0
\eeqn
  \end{enumeratenumeric}
\end{lem}

\begin{proof}
  These can be shown by direct computation.
\end{proof}

With this propagator, we can rephrase the regularized Feynman graph integral $W_{\varepsilon}^L ((\Gamma, n), \Phi)$
in the following way.

\begin{prop}
  Given decorated graph $(\Gamma, n)$ and $\Phi \in \Omega^{\bu}_c
  ((\R^{d'} \times \C^d)^{| \Gamma_0 |})$, we have the following equality:
  \begin{eqnarray*}
W_{\varepsilon}^L ((\Gamma, n), \Phi) = \int_{(\R^{d'} \times \C^d)^{| \Gamma_0 |} \times
    [\varepsilon, L]^{| \Gamma_1 |}} \prod_{e \in \Gamma_1 } \partial_{z^{h
    (e)}}^{n (e)} P_{t_e} (z^{h (e)} - z^{t
    (e)}, \bar{z}^{h (e)} - \bar{z}^{t
    (e)}, x^{h (e)} - x^{t (e)}) \wedge \Phi,
  \end{eqnarray*}
  where $t_e$ is the parameter associated with each edge $e \in \Gamma_1$.
\end{prop}

\begin{proof}
In the integrand
  \[ \prod_{e \in \Gamma_1 } \partial_{z^{h (e)}}^{n (e)} P_{t_e} (z^{h (e)} - z^{t
    (e)}, \bar{z}^{h (e)} - \bar{z}^{t
    (e)}, x^{h (e)} - x^{t (e)}) \wedge \Phi, \]
 $P_{t_e}$ appears $| \Gamma_1 |$ times and each provides at most a single $\d t_e$ component. We recall that only the top differential forms contributes the
  integral over $[\varepsilon, L]^{| \Gamma_1 |}$. As a consequence, the second term in $(\ref{schwinger propagator})$ has no contributions to the integral. By carefully tracking the $\pm$ signs, we obtain the following equality:
  \begin{eqnarray*}
 & & \int_{(\R^{d'} \times \C^d)^{| \Gamma_0 |} \times
    [\varepsilon, L]^{| \Gamma_1 |}} \prod_{e \in \Gamma_1 } \partial_{z^{h
    (e)}}^{n (e)} P_{t_e} (z^{h (e)} - z^{t
    (e)}, \bar{z}^{h (e)} - \bar{z}^{t
    (e)}, x^{h (e)} - x^{t (e)}) \wedge \Phi\\
    & = & (- 1)^{^{\frac{d + d' - 1}{2} | \Gamma_1 | (| \Gamma_1 | - 1) + |
    \Gamma_1 |}} \int_{(\R^{d'} \times \C^d)^{| \Gamma_0 |}}\prod_{e \in \Gamma_1 } \partial_{z^{h (e)}}^{n (e)}
    \til{P}_{\varepsilon, L} (z^{h (e)} - z^{t
    (e)}, \bar{z}^{h (e)} - \bar{z}^{t
    (e)}, x^{h (e)} - x^{t (e)}) \wedge \Phi\\
    & = & W_{\varepsilon}^L ((\Gamma, n), \Phi)
  \end{eqnarray*}
  We have used the fact that only the top differential forms contributes the
  integral over $[\varepsilon, L]^{| \Gamma_1 |}$.
\end{proof}

To prove the finiteness property, we need to realize Feynman graph integral as
a integral of a differential form over Schwinger spaces:

\begin{prop}
  Given decorated graph $(\Gamma, n)$ and $\Phi \in \Omega^{\bu}_c
  ((\R^{d'} \times \C^d)^{| \Gamma_0 |})$, we denote the
  integrand
  \[ \prod_{e \in \Gamma_1 } \partial_{z^{h (e)}}^{n (e)} P_{t_e}
     (z^{h (e)} - z^{t
    (e)}, \bar{z}^{h (e)} - \bar{z}^{t
    (e)}, x^{h (e)} - x^{t (e)})\wedge \Phi \]
  by $\til{W} ((\Gamma, n), \Phi)$. Then
  \[ \int_{(\R^{d'} \times \C^d)^{| \Gamma_0 |}} \til{W}
     ((\Gamma, n), \Phi)   \]
is a smooth differential form on $(0, L]^{| \Gamma_1 |}$ for any $L > 0$.
\end{prop}

\begin{proof}
  This can be proved by dominated convergence theorem easily. We will prove a
  stronger version later. See Proposition \ref{extension theorem2}.
\end{proof}

Our strategy to prove the finiteness of $W_{\varepsilon}^L ((\Gamma, n),
\Phi)$ can be described by the following two steps:
\begin{enumeratenumeric}
  \item Prove $\int_{(\R^{d'} \times \C^d)^{| \Gamma_0 |}}
  \til{W} ((\Gamma, n), \Phi)$ can be extended to a differential form over the compactification
  $\widetilde{[0, L]^{| \Gamma_1 |}}$,
  
  \item Since $\widetilde{[0, L]^{| \Gamma_1 |}}$ is compact, the integral
  over $\widetilde{[0, L]^{| \Gamma_1 |}}$ is automatically finite.
\end{enumeratenumeric}

Before starting our main analysis, we rephrase Feynman graph integrals using a convenient coordinate system.
We take note of the following facts. 
If a decorated graph $(\Gamma, n)$ have two connected components
  $(\Gamma', n | \nobracket_{\Gamma'})$ and $(\Gamma'', n |
  \nobracket_{\Gamma''})$, we have
  \[ W_{\varepsilon}^L ((\Gamma, n), \Phi' \cdot \Phi'') = \pm
     W_{\varepsilon}^L ((\Gamma', n | \nobracket_{\Gamma'}), \Phi')
     W_{\varepsilon}^L ((\Gamma'', n | \nobracket_{\Gamma''}), \Phi''), \]
  where $\Phi' \in \Omega^{\bu}_c ((\R^{d'} \times \C^d)^{|
  \Gamma'_0 |})$, $\Phi'' \in \Omega^{\ast, \ast}_c ((\R^{d'} \times \C^d)^{| \Gamma''_0 |})$.
  Furthermore, we notice that we notice that the pull back of $P_t (z - w, \zbar -
    \wbar, x-y)$ to the diagonal $\{(z,x,w,y)\in\R^{d'} \times \C^d\times\R^{d'} \times \C^d:z=w, x=y\}$ is zero, so if a decorated graph $(\Gamma, n)$ contains a self-loop, then
  \[ W_{\varepsilon}^L ((\Gamma, n), \Phi) = 0. \]
Therefore, without loss of generality, we assume that $\Gamma$ is connected without self-loops in what follows.

Recall that we denote the coordinate on the $i$th factor in 
\beqn\label{eqn:product}
(\R^{d'} \times \C^d)^{|\Gamma_0|} = (\R^{d'} \times \C^d) \times \cdots \times (\R^{d'} \times \C^d) 
\eeqn
by 
\beqn
(z^i, x^i) = (z^i_1,\ldots,z_{d}^i; x^i_1,\ldots,x^i_{d'}) .
\eeqn
Given a labeled graph $(\Gamma, n)$, we introduce the following coordinates:
\[ \left\{\begin{array}{l}
     z^i = w^i +  w^{|
     \Gamma_0 |}, x^i = q^i + 
     q^{| \Gamma_0 |} \text{\qquad} 1 \leqslant i \leqslant |
     \Gamma_0 | - 1\\
     z^{| \Gamma_0 |} = w^{| \Gamma_0 |},
     x^{| \Gamma_0 |} = q^{| \Gamma_0 |}
   \end{array}\right. \]
Here we have used the ordering of vertices of $\Gamma$.

Using lemma \ref{lem:u} and coordinates $(w^i_j ; q^i_j)$ for \eqref{eqn:product}, the integrand of
Feynman graph integral becomes:
\beqn\label{integrand}
\til{W} ((\Gamma, n), \Phi) 
   = \frac{1}{\pi^{\left( d + \frac{d'}{2} \right) | \Gamma_1 |}} e^{-
  \overset{| \Gamma_0 | - 1}{\underset{i = 1}{\sum}} \underset{e \in
  \Gamma_1}{\sum} \rho_i^e w^i \cdot u^e -
  \underset{e \in \Gamma_1}{\sum} v^e \cdot
  v^e} \prod_{e \in \Gamma_1} \left( \prod_{1 \leqslant i
  \leqslant d} (u_i^e)^{n_{i, e}} \right) 
 \d^d u^e \d^{d'} v^e \wedge \Phi,  
\eeqn
where
\beqn \label{eqn:uv}
u^e = \sum_{i = 1}^{| \Gamma_0 | - 1} \frac{1}{2 t_e}
   \rho^e_i \wbar^i,\qquad v^e = \sum_{i = 1}^{|
   \Gamma_0 | - 1} \frac{1}{2 \sqrt{t_e}} \rho^e_i q^i . 
\eeqn

The proof of the following proposition shows the utility of this change of coordinates.

\begin{prop}
  If there exists a connected subgraph $\Gamma' \subseteq \Gamma$, such that
  \[ (d + d') | \Gamma'_0 | < (d + d' - 1) | \Gamma_1' | + d + d' + 1, \]
  then $\til{W} ((\Gamma, n), \Phi) = 0$.\label{vanishing result1}
\end{prop}

\begin{proof}
  Note there is a factor $\til{W} ((\Gamma', n |  \nobracket_{\Gamma'}), 1)$
  in $\til{W} ((\Gamma, n), \Phi)$, we only need to prove $\til{W}
  ((\Gamma', n |  \nobracket_{\Gamma'}), 1) = 0$.
  
Consider the map 
\beqn
g_{\Gamma'} \colon (\R^{d'} \times \C^d)^{|
     \Gamma_0' |} \times (0, + \infty)^{| \Gamma_1' |} \to (\R^{d'} \times \C^d)^{| \Gamma_1' |}
\eeqn
defined by
\[
g_{\Gamma'} (w^i, \wbar^i, q^i, t_e) = (u^e, v^e) .
\]
Since $g_{\Gamma'}$ is a anti-holomorphic map with respect to variables $(
w^i, \wbar^i)_{i \in \Gamma'_0}$, its anti-holomorphic derivative is of the form:
  \[ \br D_{(w^i, \wbar^i,
     q^i, t_e)} g \colon T^{0,1}_{(w^i,
     \wbar^i, q^i)} (\R^{d'} \times \C^d)^{| \Gamma_0' |} \oplus T_{t_e} (0, + \infty)^{|
     \Gamma_1' |} \rightarrow T^{1,0}_{(u^e,
     v^e)} (\R^{d'} \times \C^d)^{| \Gamma_1'
     |} . \]
Here $T_\bu^{0,1}$ is the tangent space spanned by anti-holomorphic derivatives in $\wbar$ and derivatives in $x$, and $T_\bu^{1,0}$ is spanned by holomorphic derivatives in $u^e$ and $v^e$.

The following vectors belong to the kernel of
  $\br D_{(w^i, \wbar^i, q^i, t_e)} g$:
  \[ \sum_{i = 1}^{| \Gamma_0' |} \partial_{w^i_j}, \quad \sum_{i = 1}^{|
     \Gamma_0' |} \partial_{x^i_{j'}}, \quad  \sum_{i = 1}^{| \Gamma_1 |} t_e
     \partial_{t_e} + \sum_{i,k} \wbar^i_k \partial_{\wbar^i_k} + \frac{1}{2}
     \sum_{\underset{1 \leqslant k \leqslant d'}{1 \leqslant i \leqslant |
     \Gamma_0 |}} x^i_k \partial_{x^i_k}, \]
  where $1 \leqslant j \leqslant d, \text{ } 1 \leqslant j' \leqslant d'$. So
  the rank of this map is bounded above by
  \[ (d + d') | \Gamma_0' | + | \Gamma_1' | - d - d' - 1 < (d + d')
     | \Gamma_1' | . \]
  On the other hand, the dimension of $T_{(u^e,
  v^e)} (\R^{d'} \times \C^d)^{| \Gamma_1'
  |}$ is $(d + d') | \Gamma_1' |$.
  We conclude that $\br D_{(w^i, \wbar^i, q^i, t_e)} g$ is not
  surjective. 

Now, notice that $\til{W} ((\Gamma', n |  \nobracket_{\Gamma'}), 1)$
  contains a factor $g_{\Gamma'}^{\ast} \left( \prod_{e \in \Gamma'_1} \d^d u^e
  \d^{d'} v^e \right)$. Since $\prod_{e \in \Gamma'_1} \d^d u^e \d^{d'} v^e$ is a
  top holomorphic form on $(\R^{d'} \times \C^d)^{| \Gamma_1'
  |}$,
  \[ g_{\Gamma'}^{\ast} {\left( \prod_{e \in \Gamma'_1} \d^d u^e \d^{d'} v^e
     \right)_{(w^i, \wbar^i,
     x^i, t_e)}}  = 0 \]
since $\br D_{(w^i, \wbar^i, q^i, t_e)}g$ is not surjective. We conclude that $\til{W} ((\Gamma, n), \Phi) = 0.$
\end{proof}

Following \eqref{integrand}, we can rephrase the Feynman graph integral:
\begin{multline}  \label{integral}
W_{\varepsilon}^L ((\Gamma, n), \Phi) = 
\frac{(- 1)^{{d'}^2 (| \Gamma_0 | - 1)}}{\pi^{\left( d + \frac{d'}{2}
  \right) | \Gamma_1 |}} \int_{\left(w^{| \Gamma_0 |}, q^{| \Gamma_0 |}
\right)} \int_{(\R^{d'} \times \C^d)^{| \Gamma_0 | - 1}
  \times [\varepsilon, L]^{| \Gamma_1 |}} \\
  e^{- \overset{| \Gamma_0 | - 1}{\underset{i = 1}{\sum}} \underset{e
  \in \Gamma_1}{\sum} \rho_i^e w^i \cdot u^e
  - \underset{e \in \Gamma_1}{\sum} v^e \cdot
  v^e} \prod_{e \in \Gamma_1} \left( \prod_{1 \leqslant i
  \leqslant d} (u_i^e)^{n_{i, e}} \right) \d^d u^e \d^{d'} v^e \wedge \Phi, 
\end{multline}
where $\int_{\left(w^{| \Gamma_0 |}, q^{| \Gamma_0 |}
\right)}$ denotes integration over the $|\Gamma_0|$-vertex
\[ \left(w^{| \Gamma_0 |}, q^{| \Gamma_0 |}
\right) \in \R^{d'} \times \C^d . \]

Notice that the integrand in \eqref{integral} is a smooth differential form in the variables $(w^i, x^i , u^e, v^e)$. 
Our goal is to utilize a compactification of Schwinger space such that these coordinate functions extend smoothly. 
This will allow us to extend
\[ \int_{\left(w^{| \Gamma_0 |}, q^{| \Gamma_0 |}
\right)}
   \int_{(\R^{d'} \times \C^d)^{| \Gamma_0 | - 1}} \til{W}
   ((\Gamma, n), \Phi) \]
to a smooth differential form on the compactification. This will be achieved in next subsection.

\

\subsection{Finiteness of Feynman graph integrals}

In this subsection, we will prove the main theorem on finiteness of Feynman graph integrals on $\R^{d'} \times \C^d$. 
The key idea can be described as follows:
\begin{enumerate}
    \item Find a
``coordinate transformation'' of $(\R^{d'} \times \C^d)^{|
\Gamma_0 |} \times (0, + \infty)^{| \Gamma_1 |}$, such that the integrand
$\til{W} ((\Gamma, n), \Phi)$ can be expressed in terms of $(M_{\Gamma}
(t)^{- 1})^{i}_j$ , $(d_{\Gamma}  (t)^{- 1})^{e j}$ and their de Rham
differentials. We refer to appendix \ref{graph theory} for the definitions of $M_\Gamma(t)$ and $\d_\Gamma(t)$.
\item By Lemma \ref{extended functions}, $(M_{\Gamma}
(t)^{- 1})^{i}_j$ and $(d_{\Gamma}  (t)^{- 1})^{e j}$ are smooth functions over the compactified Schwinger spaces (see Definition \ref{compactified Schwinger space}). So the integrand
$\til{W} ((\Gamma, n), \Phi)$ is smooth.
\item The integral $\int_{(\R^{d'} \times \C^d)^{| \Gamma_0 | - 1}} \til{W}
   ((\Gamma, n), \Phi)$ is dominated by a Gaussian integral, so it is smooth by the dominated convergence theorem.
\end{enumerate}

We start from purely topological case where the underlying space is
$\R^{d'}$, so $d = 0$ and consider the new coordinates $ (\til{q}^i, \til{t}_e)$ defined by
\beqn\label{coordinate transformation}
q^i =
   \sum_{j = 1}^{| \Gamma_0 | - 1} (M_{\Gamma} (\til{t})^{- 1})^i_j
   \til{q}^j \; , \quad t_e = \til{t}_e^2 
\eeqn
where $1 \leqslant i \leqslant |\Gamma_0 | - 1$ and  $e \in \Gamma_1$.
This change of coordinates covers the map (see lemma \ref{square map} for a definition)
\[ t_{\tmop{square}} \left|_{(0, + \infty)^{| \Gamma_1 |}} \right. \colon
   (0, + \infty)^{| \Gamma_1 |} \rightarrow (0, + \infty)^{| \Gamma_1 |} . \]
In these new coordinates we have the following expression for $v^e$, see \eqref{eqn:uv}:
\begin{eqnarray}
v^e & = & \sum_{i = 1}^{| \Gamma_0 | - 1} \frac{1}{2
  \sqrt{t_e}} \rho^e_i q^i \nonumber\\
  & = & \sum_{i = 1}^{| \Gamma_0 | - 1} \sum_{j = 1}^{| \Gamma_0 | - 1}
  \frac{1}{2 \til{t}_e} \rho^e_i (M_{\Gamma} (\til{t})^{- 1})^i_j
  \til{q}^j \nonumber\\
  & = & \sum_{j = 1}^{| \Gamma_0 | - 1} \frac{1}{2} (d_{\Gamma} 
  (\til{t})^{- 1})^{e j} \til{q}^j \label{111} . 
\end{eqnarray}
By lemma \ref{extended functions}, the integrand
\beqn  \label{222}
\til{W}_{\tmop{top}} ((\Gamma, n), \Phi_{\tmop{top}}) \nonumber\\
   = \frac{1}{\pi^{\left( \frac{d'}{2} \right) | \Gamma_1 |}} e^{-
  \underset{e \in \Gamma_1}{\sum} v^e \cdot
  v^e} \prod_{e \in \Gamma_1} \d^{d'} v^e \wedge \Phi_{\tmop{top}} \left(
  \sum_{j = 1}^{| \Gamma_0 | - 1} (M_{\Gamma} (\til{t})^{- 1})^i_j
  \til{q}^j \; , \; q^{| \Gamma_0 |} \right) 
\eeqn
can be extended to $(\R^{d'})^{| \Gamma_0 |} \times \widetilde{\left[
0, \sqrt{L} \right]^{| \Gamma_1 |}}$, where $\Phi_{\tmop{top}} \in
\Omega^{\bu}_c ((\R^{d'})^{| \Gamma_0 |})$. 
As
$(\R^{d'})^{| \Gamma_0 | - 1}$ is non-compact, we will need to use
dominated convergence theorem to prove that the differential form
\beqn\label{eqn:pullback}
\left(t_{\tmop{square}}\left|_{(0, +
\infty)^{| \Gamma_1 |}} \right. \right)^{\ast}  \int_{(\R^{d'})^{| \Gamma_0 | - 1}}
\til{W}_{\tmop{top}} ((\Gamma, n), \Phi_{\tmop{top}})
\eeqn
can be extended to
$\R^{d'}_{q^{| \Gamma_0 |}} \times \widetilde{\left[
0, \sqrt{L} \right]^{| \Gamma_1 |}}$.
For notational convenience, we will use 
\beqn
t_{\tmop{square}}^{\ast}
\int_{(\R^{d'})^{| \Gamma_0 | - 1}} \til{W}_{\tmop{top}} ((\Gamma,
n), \Phi_{\tmop{top}})
\eeqn
in place of the expression \eqref{eqn:pullback}

\begin{lem}
  \label{bound for exponent}If we use $M_{\Gamma} (\til{t})$ to denote the
  matrix with matrix elements $M_{\Gamma} (t)^i_j$, we have the
  following inequality:
  \[ M_{\Gamma} (\til{t})^{- 1} M_{\Gamma} (\til{t}^2) M_{\Gamma}
     (\til{t})^{- 1} \geqslant \frac{1}{c_{\Gamma}} \tmop{Id}, \]
  where $c_{\Gamma}$ is a constant which only depends on the graph $\Gamma$,
  and the inequality uses the Loewner order structure of symmetric matrices.
\end{lem}

\begin{proof}
  We use $\rho$ to denote the $| \Gamma_1 | \times (| \Gamma_0 | - 1)$ matrix
  with entries $\rho^e_i$, and use \~{$t$} to denote the diagonal $| \Gamma_1
  | \times | \Gamma_1 |$ matrix with diagonal entries $\til{t}_e$. We have
  \begin{eqnarray*}
    &  & M_{\Gamma} (\til{t})^{- 1} M_{\Gamma} (\til{t}^2) M_{\Gamma}
    (\til{t})^{- 1}\\
    & = & M_{\Gamma} (\til{t})^{- 1} \rho^T \til{t}^2 \rho M_{\Gamma}
    (\til{t})^{- 1}\\
    & = & M_{\Gamma} (\til{t})^{- 1} \rho^T \til{t} \til{t} \rho
    M_{\Gamma} (\til{t})^{- 1}\\
    & \geqslant & \frac{1}{c_{\Gamma}} M_{\Gamma} (\til{t})^{- 1} \rho^T
    \til{t} \rho \rho^T \til{t} \rho M_{\Gamma} (\til{t})^{- 1}\\
    & = & \frac{1}{c_{\Gamma}} M_{\Gamma} (\til{t})^{- 1} M_{\Gamma}
    (\til{t}) M_{\Gamma} (\til{t}) M_{\Gamma} (\til{t})^{- 1}\\
    & = & \frac{1}{c_{\Gamma}} \tmop{Id} .
  \end{eqnarray*}
  We have used the following inequality:
  \[ \rho \rho^T \leqslant C_{\Gamma} \tmop{Id}, \]
  where
  \[ C_{\Gamma} = | \Gamma_1 | \underset{e, e' \in \Gamma_1}{\max} \{ | (\rho
     \rho^T)^{e e'} | \} . \]
\end{proof}

With this lemma, we can the following result:

\begin{prop}
  \label{extension theorem1} Let $(\Gamma, n)$ be a connected decorated graph 
  without self-loops, and $\Phi_{\tmop{top}} \in \Omega^{\bu}_c
  ((\R^{d'})^{| \Gamma_0 |})$.
  Then
\beqn\label{eqn:tsquare1}
t_{\tmop{square}}^{\ast}
\int_{(\R^{d'})^{| \Gamma_0 | - 1}} \til{W}_{\tmop{top}} ((\Gamma,
n), \Phi_{\tmop{top}})
\eeqn
  can be extended to a smooth form on $\R^{d'}_{q^{|
  \Gamma_0 |}} \times \widetilde{\left[ 0, \sqrt{L} \right]^{| \Gamma_1 |}}$.

Furthermore, the map
\beqn
\Omega^{\bu}_c ((\R^{d'})^{| \Gamma_0|}) \to  \Omega^{\bu}
    \left( \R^{d'}_{q^{| \Gamma_0 |}} \times
    \widetilde{\left[ 0, \sqrt{L} \right]^{| \Gamma_1 |}} \right)
\eeqn
which sends $\Phi_{\op{top}}$ to \eqref{eqn:tsquare1} is a continuous map between topological vector spaces.
\end{prop}

\begin{proof}
  By lemma \ref{extended functions} and $\left( \ref{222} \right)$, we know
  $\til{W}_{\tmop{top}} ((\Gamma, n), \Phi_{\tmop{top}})$ can be extended to
  $(\R^{d'})^{| \Gamma_0 |} \times \widetilde{\left[ 0, \sqrt{L}
  \right]^{| \Gamma_1 |}}$.
  
  Now, we note $\til{W}_{\tmop{top}} ((\Gamma, n), \Phi_{\tmop{top}})$ can
  be rewrited in the following form:
\beqn
\til{W}_{\tmop{top}} ((\Gamma, n), \Phi_{\tmop{top}}) = e^{- \frac{1}{4} \underset{i = 1}{\overset{| \Gamma_0 | - 1}{\sum}}
   \til{q}^i \cdot (M_{\Gamma} (\til{t})^{- 1}
    M_{\Gamma} (\til{t}^2) M_{\Gamma} (\til{t})^{- 1})_{i j}
    \til{q}^j} P (\til{q}^i) \omega
    \wedge \Phi_{\tmop{top}},
\eeqn
  where $\omega$ is a constant differential form and $P$ is a polynomial with coefficients
  in terms of $(M_{\Gamma} (t)^{- 1})^{i j}$, $(d_{\Gamma}^{- 1})^{e j}$ and
  their derivatives.
  
  Since $\widetilde{\left[ 0, \sqrt{L} \right]^{| \Gamma_1 |}}$ is compact,
  the values of $(M_{\Gamma} (\til{t})^{- 1})^i_j$, $(d_{\Gamma} 
  (\til{t})^{- 1})^{e j}$ and their derivatives with respect to $\{
  \til{t}_e \}$ are bounded. So if $D$ is any order $|
  D |$ differential operator on $(\R^{d'})^{| \Gamma_0 |} \times
  \widetilde{\left[ 0, \sqrt{L} \right]^{| \Gamma_1 |}}$, we have
\begin{align*} D (\til{W}_{\tmop{top}} ((\Gamma, n), \Phi_{\tmop{top}})) &= D \left( e^{- \frac{1}{4} \underset{i = 1}{\overset{| \Gamma_0 | -
    1}{\sum}} \underset{j = 1}{\overset{| \Gamma_0 | - 1}{\sum}}
    \til{q}^i \cdot (M_{\Gamma} (\til{t})^{- 1}
    M_{\Gamma} (\til{t}^2) M_{\Gamma} (\til{t})^{- 1})_{i j}
    \til{q}^j} P (\til{q}^i) \omega
    \wedge \Phi_{\tmop{top}} \right)\\
    & = e^{- \frac{1}{4} \underset{i = 1}{\overset{| \Gamma_0 | - 1}{\sum}}
    \underset{j = 1}{\overset{| \Gamma_0 | - 1}{\sum}}
    \til{q}^i \cdot (M_{\Gamma} (\til{t})^{- 1}
    M_{\Gamma} (\til{t}^2) M_{\Gamma} (\til{t})^{- 1})_{i j}
    \til{q}^j} P' (\til{q}^i) \omega
    \wedge D' \Phi_{\tmop{top}}
\end{align*}
  where $D'$ is another differential operator on $(\R^{d'})^{|
  \Gamma_0 |} \times \widetilde{\left[ 0, \sqrt{L} \right]^{| \Gamma_1 |}}$,
  $P' (\til{q}^i)$ is a polynomial with coefficients in
  terms of $(M_{\Gamma} (t)^{- 1})^{i j}$, $(d_{\Gamma}^{- 1})^{e j}$ and
  their derivatives.
  
  Finally, by lemma $\ref{bound for exponent}$, we can get
  \begin{equation}
    | D (\til{W}_{\tmop{top}} ((\Gamma, n), \Phi_{\tmop{top}})) | \leqslant
    C' e^{- \frac{1}{4 c_{\Gamma}} \underset{i = 1}{\overset{| \Gamma_0 | -
    1}{\sum}} \til{q}^i \cdot \til{q}^i}
    | P' (\til{q}^i) | \max \{ | D' \Phi_{\tmop{top}} | \},
    \label{inequality1}
  \end{equation}
  where $C'$ is some constant. Since the right hand side of $\left(
  \ref{inequality1} \right)$ is absolute integrable and independent of
  $\R^{d'}_{q^{| \Gamma_0 |}} \times
    \widetilde{\left[ 0, \sqrt{L} \right]^{| \Gamma_1 |}}$, the expression
  \beqn
t_{\tmop{square}}^{\ast}
\int_{(\R^{d'})^{| \Gamma_0 | - 1}} \til{W}_{\tmop{top}} ((\Gamma,
n), \Phi_{\tmop{top}})
\eeqn
  is smooth over $\R^{d'}_{q^{| \Gamma_0 |}} \times
    \widetilde{\left[ 0, \sqrt{L} \right]^{| \Gamma_1 |}}$ by dominated convergence theorem.
\end{proof}

Now, we can consider Feynman graph integrals on $(\R^{d'} \times \C^d)^{| \Gamma_0 |}$.

\begin{prop}
  \label{extension theorem2}Given a connected decorated graph $(\Gamma, n)$
  without self-loops, the following statements are true:
  \begin{enumerate}
    \item The map
\beqn
\Omega^{\bu}_c ((\R^{d'} \times \C^d)^{|
      \Gamma_0 |}) \to \Omega^{\bu} \left( (\R^{d'} \times \C^d)_{\left( w^{| \Gamma_0 |},
      q^{| \Gamma_0 |} \right)} \times \widetilde{\left[ 0,
      \sqrt{L} \right]^{| \Gamma_1 |}} \right)
\eeqn
    which sends $\Phi$ to the differential form
    \beqn
    t_{\tmop{square}}^{\ast} \int_{(\R^{d'} \times \C^d)^{| \Gamma_0 | - 1}} \til{W} ((\Gamma, n), \Phi)
      \eeqn
    is a continuous map between topological vector spaces.
    
    \item The map
\beqn
\Omega^{\bu}_c ((\R^{d'} \times \C^d)^{|
      \Gamma_0 |}) \to \Omega^{\bu} \left( \widetilde{\left[ 0, \sqrt{L} \right]^{|
      \Gamma_1 |}} \right)
\eeqn
    which sends $\Phi$ to the differential form
    \beqn
    t_{\tmop{square}}^{\ast} \int_{(\R^{d'} \times \C^d)_{\left( w^{| \Gamma_0 |},
      q^{| \Gamma_0 |} \right)}} \int_{(\R^{d'} \times \C^d)^{| \Gamma_0 | - 1}} \til{W} ((\Gamma, n), \Phi)
      \eeqn
    is also a continuous map between topological vector spaces.
  \end{enumerate}
\end{prop}

\begin{proof}
  
  \begin{enumerate}
    \item Consider first the case $d' = 0$.
    In this case, it is shown in \cite{wang2024feynman} that the map
    \beqn
\Omega^{\bu}_c ((\C^d)^{| \Gamma_0|}) \to 
\Omega^{\bu} \left( \C^d_{w^{|
      \Gamma_0 |}} \times \widetilde{\left[ 0, \sqrt{L} \right]^{| \Gamma_1
      |}} \right) .
\eeqn
which sends $\Phi_{hol}$ to 
\beqn
 t_{\tmop{square}}^{\ast} \int_{(\C^d)^{|
      \Gamma_0 | - 1}} \til{W}_{hol} ((\Gamma, n),
      \Phi_{hol})
      \eeqn
is continuous.

Now, we consider the general case. We only need to prove the Feynman graph integral defines a continuous map from \[\Omega^{\bu} (K \times K')\] to \[\Omega^{\bu} \left( (\R^{d'} \times \C^d)_{\left( w^{| \Gamma_0 |},
      q^{| \Gamma_0 |} \right)} \times \widetilde{\left[ 0,
      \sqrt{L} \right]^{| \Gamma_1 |}} \right),\]
      where $K\subset (\C^d)^{| \Gamma_0 |}$ and $K'\subset (\R^d)^{| \Gamma_0 |}$ are compact subsets. 

    We notice the following K\"unneth isomorphisms between topological vector spaces (see \cite{treves2006topological}):
    \begin{enumeratenumeric}
        \item \[ \Omega^{\bu} (K\times K') 
       \cong \Omega^{\bu}(K) \hat{\otimes}
       \Omega^{\bu}(K'), \]
       where $\hat{\otimes}$ is the complete projective tensor product.
       \item \begin{eqnarray*}
      &  & \Omega^{\bu} \left( (\R^{d'} \times \C^d)_{\left(
      w^{| \Gamma_0 |}, q^{| \Gamma_0 |}
      \right)} \times \widetilde{\left[ 0, \sqrt{L} \right]^{| \Gamma_1 |}}
      \times \widetilde{\left[ 0, \sqrt{L} \right]^{| \Gamma_1 |}} \right)\\
      & = & \Omega^{\bu} \left( \C^d_{w^{| \Gamma_0
      |}} \times \widetilde{\left[ 0, \sqrt{L} \right]^{| \Gamma_1 |}}
      \right) \hat{\otimes} \Omega^{\bu} \left(
      \R^{d'}_{q^{| \Gamma_0 |}} \times
      \widetilde{\left[ 0, \sqrt{L} \right]^{| \Gamma_1 |}} \right).
    \end{eqnarray*}
    \end{enumeratenumeric}

    Combine Proposition \ref{extension theorem1} and the above isomorphisms, we
    have the continuous map between topological vector
    spaces
    \begin{eqnarray*}
\Omega^{\bu}(K\times K') \to  \Omega^{\bu} \left( (\R^{d'} \times \C^d)_{\left( w^{| \Gamma_0 |},
      q^{| \Gamma_0 |} \right)} \times \widetilde{\left[ 0,
      \sqrt{L} \right]^{| \Gamma_1 |}} \times \widetilde{\left[ 0, \sqrt{L}
      \right]^{| \Gamma_1 |}} \right) ,
    \end{eqnarray*}
    which sends $\Phi_{hol} \otimes \Phi_{top}$ to 
    \beqn
    t_{\tmop{square}}^{\ast} \int_{(\C^d)^{|
      \Gamma_0 | - 1}} \til{W}_{\tmop{hol}} ((\Gamma, n),
      \Phi_{\tmop{hol}})
       \otimes t_{\tmop{square}}^{\ast} \int_{(\R^{d'})^{|
      \Gamma_0 | - 1}} \til{W}_{\tmop{top}} ((\Gamma, n),
      \Phi_{\tmop{top}}).
\eeqn
Composition with restriction along the diagonal yields the desired map.
    
    \item     Note that integration over $(\R^{d'} \times \C^d)_{\left(
    w^{| \Gamma_0 |}, q^{| \Gamma_0 |}
    \right)}$ is a continuous
    \beqn
    \Omega^{\bu}_c \left( (\R^{d'} \times \C^d)_{\left( w^{| \Gamma_0 |},
    q^{| \Gamma_0 |} \right)} \times \widetilde{\left[ 0,
    \sqrt{L} \right]^{| \Gamma_1 |}} \right) \to \Omega^{\bu} \left(
    \widetilde{\left[ 0, \sqrt{L} \right]^{| \Gamma_1 |}} \right)
    \eeqn
    Since the support of $t_{\tmop{square}}^{\ast}
    \int_{(\R^{d'})^{| \Gamma_0 | - 1}} \til{W}_{\tmop{top}}
    ((\Gamma, n), \Phi_{\tmop{top}})$ is compact, our claim follows from
    proposition \ref{extension theorem1}.
  \end{enumerate}
\end{proof}

We arrive at the main result and a strengthening of theorem \ref{thm:renormalization}.

\begin{cor}
  \label{extension theorem3}Given a connected decorated graph $(\Gamma, n)$
  without self-loops, and $\Phi \in \Omega^{\bu}_c ((\R^{d'} \times \C^d)^{| \Gamma_0 |})$, we have $W_0^L ((\Gamma, n), \Phi) < +
  \infty$. Furthermore, the map
  \[ \Omega^{\bu}_c ((\R^{d'} \times \C^d)^{|
     \Gamma_0 |}) \rightarrow \C \]
  which sends $\Phi$ to $W_0^L ((\Gamma, n), \Phi)$ defines a distribution on $(\R^{d'} \times \C^d)^{| \Gamma_0
  |}$.
\end{cor}
\begin{proof}
Integration over $\widetilde{\left[ 0, \sqrt{L} \right]^{| \Gamma_1 |}}$
  is a continuous map
  \beqn
  \Omega^{\bu} \left( \widetilde{\left[ 0, \sqrt{L}
  \right]^{| \Gamma_1 |}} \right) \to \C .
  \eeqn
  The result is then a consequence of proposition \ref{extension theorem2}.
\end{proof}

\subsection{Anomaly integrals}

In this subsection, we will prove that the quantum master equation is satisfied when the spacetime $\R^{d'} \times \C^d$ satisfies $d'>1$. The idea can be described as follows:
\begin{enumerate}
    \item By Corollary \ref{weight QME}, the quantum master equation can be verified if $(\ref{weight QME2})$ holds for any stable graph $\gamma$. We observe that the first term in $(\ref{weight QME2})$ has the following form (after integration by parts):\[
    \sum_{n\in dec(\Gamma)}W_0^L ((\Gamma, n), (\bar{\partial} + \d_{deRham})\Phi_{n}),
    \]
    where $\Gamma$ is a directed graph whose underlying undirected graph is obtained from $\gamma$ by removing all external edges. Here, $dec(\Gamma)$ is the set of decorations of $\Gamma$, and $\Phi_{n}\in \Omega^{\bu}_c ((\C^d \times
  \R^{d'})^{| \Gamma_0 |})$ is nonzero for only finitely many $n\in dec(\Gamma)$.
  \item We show that $ (\bar{\partial} + \d_{deRham}) W_0^L ((\Gamma, n), -)$ equals some integrals over the boundaries of compactified Schwinger spaces. Moreover, we will give some precise descriptions of these integrals.
  \item The integrals over some parts of boundaries of compactified Schwinger spaces will cancel the second term and the third term in $(\ref{weight QME2})$, and we will prove that the integrals over other parts (which are called anomaly integrals) are zero.
\end{enumerate}

\subsubsection{Integrals over boundaries of compactified Schwinger spaces.}

\begin{prop}
  \label{differential transfer}Given a connected decorated graph $(\Gamma, n)$
  without self-loops, and $\Phi \in \Omega^{\bu}_c ((\C^d \times
  \R^{d'})^{| \Gamma_0 |})$, $L > 0$, the following equality holds:
  \begin{eqnarray}
    &  & ((\bar{\partial} + \d_{deRham}) W_0^L) ((\Gamma, n), -) \nonumber\\
    & = & (- 1)^{| \Gamma_1 |}  \int_{\partial
    \widetilde{\left[ 0, \sqrt{L} \right]^{| \Gamma_1 |}}}t_{\tmop{square}}^{\ast} \int_{(\R^{d'} \times \C^d)^{| \Gamma_0 |}} \til{W} ((\Gamma, n), -) .
    \label{QME} 
  \end{eqnarray}
\end{prop}

\begin{proof}
  By Lemma \ref{closeness of propagator} and Leibniz's rule, we have
  \begin{eqnarray*}
     & &t_{\tmop{square}}^{\ast} \til{W} ((\Gamma, n), (\bar{\partial} +
    \d_{deRham}) \Phi)\\ & = & (- 1)^{(d + d' - 1) | \Gamma_1 |} (\bar{\partial} + \d_{deRham}) \left( \prod_{e \in \Gamma_1 } \partial_{z^{h (e)}}^{n (e)}
    P_{\til{t}_e^2} (z^{h (e)} - z^{t
    (e)}, \bar{z}^{h (e)} - \bar{z}^{t
    (e)}, x^{h (e)} - x^{t (e)}) \wedge
    \Phi \right)\\
    & + & (- 1)^{(d + d' - 1) | \Gamma_1 |} d_{\til{t}} (\til{W}
    ((\Gamma, n), \Phi))
  \end{eqnarray*}
  therefore, by Stokes' theorem, we have
  \begin{eqnarray*}
& &((\bar{\partial} + \d_{deRham}) W_0^L) ((\Gamma, n), \Phi)\\
    & = & {(- 1)^{(d + d') | \Gamma_1 | + d' | \Gamma_0 |}}  W_0^L ((\Gamma,
    n), (\bar{\partial} + \d_{deRham}) \Phi)\\
    & = & (- 1)^{| \Gamma_1 |} \int_{\widetilde{\left[ 0, \sqrt{L} \right]^{|
    \Gamma_1 |}}} d_{\til{t}} t_{\tmop{square}}^{\ast} \int_{(\R^{d'} \times \C^d)^{| \Gamma_0 |}} \til{W} ((\Gamma, n), \Phi)\\
    & = & (- 1)^{| \Gamma_1 |} \int_{\partial \widetilde{\left[ 0, \sqrt{L}
    \right]^{| \Gamma_1 |}}} t_{\tmop{square}}^{\ast} \int_{(\R^{d'} \times \C^d)^{| \Gamma_0 |}} \til{W} ((\Gamma, n), \Phi) .
  \end{eqnarray*}
  
\end{proof}

We have shown that $(\bar{\partial} + \d_{deRham}) W_0^L ((\Gamma, n), -)$ is given by
integrals over the boundaries of Schwinger spaces. By Proposition \ref{boundary
description}, the boundaries $\partial \widetilde{\left[ 0, \sqrt{L}
\right]^{| \Gamma_1 |}}$ has the following decomposition:
\[ \partial \widetilde{\left[ 0, \sqrt{L} \right]^{| \Gamma_1 |}} = \left( -
   \partial_0 \widetilde{\left[ 0, \sqrt{L} \right]^{| \Gamma_1 |}} \right)
   \cup \partial_{\sqrt{L}} \widetilde{\left[ 0, \sqrt{L} \right]^{| \Gamma_1
   |}}, \]
where $\partial_0 \widetilde{\left[ 0, \sqrt{L} \right]^{| \Gamma_1
|}}$($\partial_{\sqrt{L}} \widetilde{\left[ 0, \sqrt{L} \right]^{| \Gamma_1
|}}$) describe the boundary components near the origin(away from the origin).
More precisely, we have
\[ \left\{\begin{array}{l}
     \partial_0 \widetilde{\left[ 0, \sqrt{L} \right]^{| \Gamma_1 |}} =
     \bigcup_{\Gamma' \subseteq \Gamma} (- 1)^{\sigma (\Gamma'_1, \Gamma_1 \backslash
     \Gamma_1')} \widetilde{(0, + \infty)^{| \Gamma_1' |}} /\R^+ \times
     \widetilde{\left[ 0, + \sqrt{L} \right]^{| \Gamma_1 \backslash \Gamma_1'
     |}}\\
     \partial_{\sqrt{L}} \widetilde{\left[ 0, \sqrt{L} \right]^{| \Gamma_1 |}}
     = \bigcup_{e \in \Gamma_1} (- 1)^{| e |} \left\{ \sqrt{L} \right\} \times
     \widetilde{\left[ 0, \sqrt{L} \right]^{| (\Gamma / e)_1 |}}
   \end{array}\right. . \]
The integrals over $\partial_{\sqrt{L}} \widetilde{\left[ 0, \sqrt{L}
\right]^{| \Gamma_1 |}}$ reflects global aspects of Feynman graph integrals,
and it will not affect the construction of factorization algebras on
$\R^{d'} \times \C^d$. The integrals over $\partial_0
\widetilde{\left[ 0, \sqrt{L} \right]^{| \Gamma_1 |}}$ reflects local aspects
of Feynman graph integrals, and we need some vanishing results for such
integrals to construct factorization algebras. 

\subsubsection{Anomaly integrals and Laman graphs}

\

In this subsection, we will describe integrals over $\partial_0
\widetilde{\left[ 0, \sqrt{L} \right]^{| \Gamma_1 |}}$ in details and relate
them to Laman graphs (see Definition \ref{Laman graph1}). Given a subgraph
$\Gamma' \subseteq \Gamma$, we notice the following equality:
\begin{eqnarray*}
  &  & \int_{\widetilde{(0, + \infty)^{| \Gamma_1' |}} /\R^+ \times
  \widetilde{\left[ 0, + \sqrt{L} \right]^{| \Gamma_1 \backslash \Gamma_1'
  |}}} t_{\tmop{square}}^{\ast} \int_{(\R^{d'} \times \C^d)^{|
  \Gamma_0 |}} \til{W} ((\Gamma, n), -)\\
  & = & \int_{\partial C_{\Gamma'_1} \cap \widetilde{\left[ 0, \sqrt{L}
  \right]^{| \Gamma_1 |}}} t_{\tmop{square}}^{\ast} \int_{(\R^{d'} \times \C^d)^{| \Gamma_0 |}} \til{W} ((\Gamma, n), -) .
\end{eqnarray*}
This equality will provide us a way to perform concrete computations. We first
concentrate on the case $\Gamma' = \Gamma$. In this case, $\partial C_{\Gamma_1}
\cap \widetilde{\left[ 0, \sqrt{L} \right]^{| \Gamma_1 |}} = \partial
C_{\Gamma_1}$. In the following, we use $O_{(\Gamma, n)}$ to denote the
following integral:
\[ O_{(\Gamma, n)} (\Phi) = \int_{\partial C_{\Gamma_1}}
   t_{\tmop{square}}^{\ast} \int_{(\R^{d'} \times \C^d)^{|
   \Gamma_0 |}} \til{W} ((\Gamma, n), \Phi), \text{\quad for } \Phi \in
   \Omega^{\bu}_c ((\R^{d'} \times \C^d)^{| \Gamma_0 |}) . \]
\begin{prop}
  \label{laman anomaly}Given a decorated graph\footnote{We do not assume
  $\Gamma$ is connected here.} $(\Gamma, n)$ without self-loops, and $\Phi$ is a compactly supported differential form on $(\R^{d'} \times \C^d)^{| \Gamma_0 |})$. The
  following statements hold:
  \begin{enumeratenumeric}
    \item If $\Gamma$ is not a Laman graph (see Definition \ref{Laman
    graph1}), we have
    \[ O_{(\Gamma, n)} (\Phi) = \int_{\partial C_{\Gamma_1}}
       t_{\tmop{square}}^{\ast} \int_{(\R^{d'} \times \C^d)^{|
       \Gamma_0 |}} \til{W} ((\Gamma, n), \Phi) = 0. \]
    \item If $\Gamma$ is a Laman graph,
    \begin{eqnarray}
      O_{(\Gamma, n)} (\Phi) & = & \int_{\partial C_{\Gamma_1}}
      t_{\tmop{square}}^{\ast} \int_{(\R^{d'} \times \C^d)^{|
      \Gamma_0 |}} \til{W} ((\Gamma, n), \Phi) \nonumber\\
      & = & \int_{(\R^{d'} \times \C^d)_{\left(
     w^{| \Gamma_0 |}, q^{| \Gamma_0 |}
      \right)}} (D \Phi) \left|_{w^i = 0,
      q^i = 0, i \neq | \Gamma_0 |} \right., 
    \end{eqnarray}
    where $D$ is a differential operator with constant coefficients, which
    only involves holomorphic derivatives. The order of $D$ is less than or
    equal to
    \[ \underset{e \in \Gamma_1, 1 \leqslant i \leqslant d}{\overset{}{\sum}}
       n_{i, e} + | \Gamma_1 | - 1. \]
  \end{enumeratenumeric}
\end{prop}

\begin{proof}
  Let's assume $\Gamma$ has $m$ connected components $\Gamma^1, \Gamma^2,
  \ldots, \Gamma^m$. If some $\Gamma^i$ does not satisfy $\left( \ref{Laman
  condition1} \right)$, by Proposition \ref{vanishing result1}, we have
  $\til{W} ((\Gamma, n), \Phi) = 0$. Let's assume all $\Gamma^i$ satisfy
  $\left( \ref{Laman condition1} \right)$, then we have
  \[ (d + d') | \Gamma_0 | \geqslant (d + d' - 1) | \Gamma_1 | + m (d + d' +
     1) . \]
  In this case, $M_{\Gamma} (\til{t}^2)$ is not invertible, but we can
  easily know that it has exactly $m - 1$ zero eigenvectors by studying the
  weighted Laplacian of each subgraph. For each connected component, we choose a base vertex and choose an ordering such that the last $m$ vertices are base vertices. 

For each connected component $\Gamma^{j}$, assume the base vertex is $v_{j}$. we use the following coordinates:
  \[ \left\{\begin{array}{l}
     z^i = w^i +  w^{v_{j}}, x^i = q^i + 
     q^{v_{j}} \text{\qquad} i \in \Gamma^{j} \\
     z^{v_{j}} = w^{v_{j}},
     x^{v_{j}} = q^{v_{j}}
   \end{array}\right. \]
  We still use $\det (M_{\Gamma}
  (\til{t}^2))$ to denote the product of its nonzero eigenvalues.
  
  Note $t_{\tmop{square}}^{\ast} \til{W} ((\Gamma, n), \Phi)$ has the
  following form:
  \begin{eqnarray*}
    &  & e^{- \overset{| \Gamma_0 | - m}{\underset{i = 1}{\sum}} \underset{j
    = 1}{\overset{| \Gamma_0 | - m}{\sum}} \left( \frac{1}{2}
    w^i \cdot M_{\Gamma} (\til{t}^2)_{i j}
    \wbar^j + \frac{1}{4} q^i \cdot
    M_{\Gamma} (\til{t}^2)_{i j} q^j \right)} \times\\
    &  & \left( \overset{d (| \Gamma_0 | - m)}{\sum_{j = 0}}
    \underset{\underset{0 \leqslant (d + d') | \Gamma_1 | - j - k \leqslant |
    \Gamma_1 |}{0 \leqslant k \leqslant d' (| \Gamma_0 | - m),}}{\sum} P_{j,
    k} (\wbar, x, d \wbar, d
    \vec{x}) \cdot \frac{\widetilde{\til{P}}_{j, k} (d
    \til{t})}{\til{P}_{j, k} ( \til{t})} \right) \wedge \Phi,
  \end{eqnarray*}
  where $P_{j, k} (\wbar, x, d
  \wbar, d x)$ is a homogenous polynomial with
  variables $\{ \wbar_i, x, d
  \wbar_i, d x_i \}_{i \in \Gamma_0}$, its degree
  with respect to $\wbar, x, d
  \wbar, d x$ are
  \[ \underset{e \in \Gamma_1, 1 \leqslant i \leqslant d}{\overset{}{\sum}}
     n_{i, e} + d | \Gamma_1 | - j , \quad d' | \Gamma_1 | - k, \quad j, \quad k \]
  respectfully. $\til{P}_{j, k} ( \til{t})$ is a homogenous polynomial with
  variables $\{ \til{t}_e \}_{e \in \Gamma_1}$, its degree is
  \[ 2 \left( \underset{e \in \Gamma_1, 1 \leqslant i \leqslant
     d}{\overset{}{\sum}} n_{i, e} \right) + (3 d + 2 d') | \Gamma_1 | - j -
     k. \]
  $\widetilde{\til{P}}_{j, k} (d \til{t})$ is a homogenous polynomial with
  variables $\{ d \til{t}_e \}_{e \in \Gamma_1}$, its degree is
  \[ (d + d') | \Gamma_1 | - j - k. \]
  Using coordinates ${\left\{ \til{\rho} = \sqrt{\underset{e \in
  \Gamma}{\sum} \til{t}_e^2}, \til{\xi}_e =
  \frac{\til{t}_e}{\sqrt{\underset{e \in \Gamma}{\sum} \til{t}_e^2}}
  \right\}_{e \in \Gamma_1}} $, we can rewrite $t_{\tmop{square}}^{\ast}
  \til{W} ((\Gamma, n), \Phi)$ as
  \begin{eqnarray*}
    &  & e^{- \frac{1}{\til{\rho}^2} \overset{| \Gamma_0 | - m}{\underset{i
    = 1}{\sum}} \underset{j = 1}{\overset{| \Gamma_0 | - m}{\sum}} \left(
    \frac{1}{2} w^i \cdot M_{\Gamma} (\til{\xi}^2)_{i j}
    \wbar^j + \frac{1}{4} q^i \cdot
    M_{\Gamma} (\til{\xi}^2)_{i j} q^j \right)} \times\\
    &  & \til{\rho}^{- 2 \left( \underset{e \in \Gamma_1, 1 \leqslant i
    \leqslant d}{\overset{}{\sum}} n_{i, e} \right) - (3 d + 2 d') | \Gamma_1
    | + j + k} \times\\
    &  & \left( \overset{d (| \Gamma_0 | - m)}{\sum_{j = 0}}
    \underset{\underset{0 \leqslant (d + d') | \Gamma_1 | - j - k \leqslant |
    \Gamma_1 |}{0 \leqslant k \leqslant d' (| \Gamma_0 | - m),}}{\sum} P_{j,
    k} (\wbar, x, d \wbar, d
    x) \cdot \frac{\widetilde{\til{P}}_{j, k} (d (\til{\rho}
    \til{\xi}))}{\til{P}_{j, k} ( \til{\xi})} \right) \wedge \Phi,
  \end{eqnarray*}

  Note $d \til{\rho} |_{\partial C_{\Gamma_1}} = 0 \nobracket$, we have
  \begin{eqnarray*}
    & & \left( t_{\tmop{square}}^{\ast} \int_{(\R^{d'} \times \C^d)^{| \Gamma_0 |}} \til{W} ((\Gamma, n), \Phi) \right) |
    \nobracket_{\til{\rho} = 0}
    = \left( \til{\rho}^{- 2 \left( \underset{e \in \Gamma_1, 1
    \leqslant i \leqslant d}{\overset{}{\sum}} n_{i, e} \right) - (2 d + d') |
    \Gamma_1 |} \times \right.\\
    &  & \int_{(\R^{d'} \times \C^d)^{| \Gamma_0 |}} e^{-
    \frac{1}{\til{\rho}^2} \overset{| \Gamma_0 | - m}{\underset{i =
    1}{\sum}} \underset{j = 1}{\overset{| \Gamma_0 | - m}{\sum}} \left(
    \frac{1}{2} w^i \cdot M_{\Gamma} (\til{\xi}^2)_{i j}
    \wbar^j + \frac{1}{4} q^i \cdot
    M_{\Gamma} (\til{\xi}^2)_{i j} q^j \right)} \times\\
    &  & \left. \left( \overset{d (| \Gamma_0 | - m)}{\sum_{j = 0}}
    \underset{\underset{0 \leqslant (d + d') | \Gamma_1 | - j - k \leqslant |
    \Gamma_1 |}{0 \leqslant k \leqslant d' (| \Gamma_0 | - m),}}{\sum} P_{j,
    k} (\wbar, x, d \wbar, d
    \vec{x}) \cdot \frac{\widetilde{\til{P}}_{j, k} (d
    (\til{\xi}))}{\til{P}_{j, k} ( \til{\xi})} \right) \wedge \Phi
    \right) | \nobracket_{\til{\rho} = 0} .
  \end{eqnarray*}
  
  We have the following asymptotic formula for Gaussian type integral(see
  {\cite[Section 2.3]{Duistermaat2011}} for a proof) when $\til{\rho}
  \rightarrow 0$:
  
  \begin{eqnarray*}
   & & \til{\rho}^{- 2 \left( \underset{e \in \Gamma_1, 1 \leqslant i
    \leqslant d}{\overset{}{\sum}} n_{i, e} \right) - (2 d + d') | \Gamma_1 |}
    \int_{(\R^{d'} \times \C^d)^{| \Gamma_0 | - m}} e^{-
    \frac{1}{\til{\rho}^2} \overset{| \Gamma_0 | - m}{\underset{i =
    1}{\sum}} \underset{j = 1}{\overset{| \Gamma_0 | - m}{\sum}} \left(
    \frac{1}{2} w^i \cdot M_{\Gamma} (\til{\xi}^2)_{i j}
    \wbar^j + \frac{1}{4}q^i \cdot
    M_{\Gamma} (\til{\xi}^2)_{i j} q^j \right)} \times\\
    &  & \left( \overset{d (| \Gamma_0 | - m)}{\sum_{j = 0}}
    \underset{\underset{0 \leqslant (d + d') | \Gamma_1 | - j - k \leqslant |
    \Gamma_1 |}{0 \leqslant k \leqslant d' (| \Gamma_0 | - m),}}{\sum} P_{j,
    k} (\wbar, \vec{x}, d \wbar, d
    \vec{x}) \cdot \frac{\widetilde{\til{P}}_{j, k} (d
    (\til{\xi}))}{\til{P}_{j, k} ( \til{\xi})} \right) \wedge \Phi\\
    & \sim & (2 \pi)^{\left( d + \frac{1}{2} d' \right) (| \Gamma_0 | - m)}
    \til{\rho}^{- 2 \left( \underset{e \in \Gamma_1, 1 \leqslant i \leqslant
    d}{\overset{}{\sum}} n_{i, e} \right) - (2 d + d') | \Gamma_1 |} \times
    \til{\rho}^{(2 d + d') | \Gamma_0 | - m (2 d + d')}\\
    &  & \frac{1}{\det (M_{\Gamma} (\til{\xi}^2))^{d + \frac{1}{2} d'}}
    e^{\til{\rho}^2 \overset{| \Gamma_0 | - m}{\underset{i = 1}{\sum}}
    \underset{j = 1}{\overset{| \Gamma_0 | - m}{\sum}} \left(
    \overset{d}{\underset{k = 1}{\sum}} \frac{1}{2} M_{\Gamma}
    (\til{\xi}^2)^{- 1}_{i j} \partial_{w^i_k} \partial_{\til{w}_k^j} +
    \overset{d'}{\underset{k = 1}{\sum}} M_{\Gamma} (\til{\xi}^2)^{- 1}_{i
    j} \partial_{q^i_k} \partial_{q_k^j} \right)} \iota_{\omega^{\vee}}\\
    &  & \left( \left( \overset{d (| \Gamma_0 | - m)}{\sum_{j = 0}}
    \underset{\underset{0 \leqslant (d + d') | \Gamma_1 | - j - k \leqslant |
    \Gamma_1 |}{0 \leqslant k \leqslant d' (| \Gamma_0 | - m),}}{\sum} P_{j,
    k} (\wbar, \vec{x}, d \wbar, d
    \vec{x}) \cdot \frac{\widetilde{\til{P}}_{j, k} (d
    (\til{\xi}))}{\til{P}_{j, k} ( \til{\xi})} \right) \wedge \Phi
    \right)\\
    &  & \left|_{w^i = 0, q^i = 0, i \leqslant |
    \Gamma_0 |-m} \right.\\
    & = & (2 \pi)^{\left( d + \frac{1}{2} d' \right) (| \Gamma_0 | - m)}
    \til{\rho}^{- 2 \left( \underset{e \in \Gamma_1, 1 \leqslant i \leqslant
    d}{\overset{}{\sum}} n_{i, e} \right) + (2 d + d') (| \Gamma_0 | - |
    \Gamma_1 |) - m (2 d + d')} \times\\
    &  & \frac{1}{\det (M_{\Gamma} (\til{\xi}^2))^{d + \frac{1}{2} d'}}
    \sum_{l = \left( \underset{e \in \Gamma_1, 1 \leqslant i \leqslant
    d}{\overset{}{\sum}} n_{i, e} \right) + d | \Gamma_1 | - j + \frac{1}{2}
    (d' | \Gamma_1 | - k)}^{+ \infty}\\
    &  & \frac{1}{l!} \til{\rho}^{2 l} \left( \overset{| \Gamma_0 | -
    m}{\underset{i = 1}{\sum}} \underset{j = 1}{\overset{| \Gamma_0 | -
    m}{\sum}} \left( \overset{d}{\underset{k = 1}{\sum}} \frac{1}{2}
    M_{\Gamma} (\til{\xi}^2)^{- 1}_{i j} \partial_{w^i_k}
    \partial_{\til{w}_k^j} + \overset{d'}{\underset{k = 1}{\sum}} M_{\Gamma}
    (\til{\xi}^2)^{- 1}_{i j} \partial_{q^i_k} \partial_{q_k^j} \right)
    \right)^l \iota_{\omega^{\vee}}\\
    &  & \left( \left( \overset{d (| \Gamma_0 | - m)}{\sum_{j = 0}}
    \underset{\underset{0 \leqslant (d + d') | \Gamma_1 | - j - k \leqslant |
    \Gamma_1 |}{0 \leqslant k \leqslant d' (| \Gamma_0 | - m),}}{\sum} P_{j,
    k} (\wbar, \vec{x}, d \wbar, d
    \vec{x}) \cdot \frac{\widetilde{\til{P}}_{j, k} (d
    (\til{\xi}))}{\til{P}_{j, k} ( \til{\xi})} \right) \wedge \Phi
    \right)\\
    &  & \left|_{w^i = 0, q^i = 0, i \leqslant
    | \Gamma_0 |-m} \right.,
  \end{eqnarray*}
  where $\iota_{\omega^{\vee}}$ is the interior product (or the contraction) of a differential form with a
  nonzero constant top polyvector field $\omega^{\vee}$ on $(\R^{d'} \times \C^d)^{| \Gamma_0 | - m}$, i.e. it maps a volume form to
  its coefficient.
  
  We notice that for the top form part of above formula, the lowest power of
  $\til{\rho}$ is
  \begin{eqnarray*}
    &  & (2 d + d') (| \Gamma_0 | - | \Gamma_1 |) - m (2 d + d') + 2 d |
    \Gamma_1 | - 2 j + d' | \Gamma_1 | - k\\
    & = & (2 d + d') | \Gamma_0 | - m (2 d + d') - 2 j - k\\
    & = & (2 d + d') | \Gamma_0 | - m (2 d + d') - j + | \Gamma_1 | - 1 - (d
    + d') | \Gamma_1 |\\
    & \geqslant & (2 d + d') | \Gamma_0 | - m (2 d + d') - d (| \Gamma_0 | -
    m) + | \Gamma_1 | - 1 - (d + d') | \Gamma_1 |\\
    & = & (d + d') | \Gamma_0 | - (d + d' - 1) | \Gamma_1 | - m (d + d') - 1\\
    &\geqslant & (d + d') | \Gamma_0 | - (d + d' - 1) | \Gamma_1 | - m (d + d') - m\\
    &\geqslant& 0,
  \end{eqnarray*}
  where the equality holds only when $m=1$ and the last inequality is an equality. So if $\Gamma$ is not a (connected) Laman graph,
  \[ \int_{\partial C_{\Gamma_1}} t_{\tmop{square}}^{\ast} \int_{(\R^{d'} \times \C^d)^{| \Gamma_0 |}} \til{W} ((\Gamma, n), \Phi) =
     0. \]

  When $\Gamma$ is indeed a Laman graph, we have
  \[ \int_{\partial C_{\Gamma_1}} t_{\tmop{square}}^{\ast} \int_{(\R^{d'} \times \C^d)^{| \Gamma_0 |}} \til{W} ((\Gamma, n), \Phi) =
     \int_{(\R^{d'} \times \C^d)_{\left( w^{|
     \Gamma_0 |}, q^{| \Gamma_0 |} \right)}} (D \Phi)
     \left|_{w^i = 0, i \neq | \Gamma_0 |} \right., \]
  where
  \begin{eqnarray*}
    D & = & (2 \pi)^{\left( d + \frac{1}{2} d' \right) (| \Gamma_0 | - m)}
    \int_{\partial C_{\Gamma_1}} \left( \frac{1}{\det (M_{\Gamma}
    (\til{\xi}^2))^{d + \frac{1}{2} d'} l_0 !} \right.\\
    &  & \left( \overset{| \Gamma_0 | - 1}{\underset{i = 1}{\sum}}
    \underset{j = 1}{\overset{| \Gamma_0 | - 1}{\sum}} \left(
    \overset{d}{\underset{k = 1}{\sum}} \frac{1}{2} M_{\Gamma}
    (\til{\xi}^2)^{- 1}_{i j} \partial_{w^i_k} \partial_{\til{w}_k^j} +
    \overset{d'}{\underset{k = 1}{\sum}} M_{\Gamma} (\til{\xi}^2)^{- 1}_{i
    j} \partial_{q^i_k} \partial_{q_k^j} \right) \right)^{l_0}
    \iota_{\omega^{\vee}}\\
    &  & \left( \left( P_{d (| \Gamma_0 | - 1), d' (| \Gamma_0 | - 1)}
    (\wbar, x, d \wbar, d x)
    \cdot \frac{\widetilde{\til{P}}_{d (| \Gamma_0 | - 1), d' (|
    \Gamma_0 | - 1)} (d (\til{\xi}))}{\til{P}_{d (| \Gamma_0 | - 1), d' (|
    \Gamma_0 | - 1)} ( \til{\xi})} \right) \wedge \Phi \right) \left|_{w^i = 0, q^i = 0, i \neq |
    \Gamma_0 |} \right.,
  \end{eqnarray*}
  where
  \[ l_0 = \left( \underset{e \in \Gamma_1, 1 \leqslant i \leqslant
     d}{\overset{}{\sum}} n_{i, e} \right) + \left( d + \frac{1}{2} d' \right)
     (| \Gamma_1 | - | \Gamma_0 | + 1) . \]
  In above formula of $D$, we can argue any terms that contain
  $\partial_{q^i_k} \Phi$ is zero: the order of
  \[ \left( \overset{| \Gamma_0 | - 1}{\underset{i = 1}{\sum}} \underset{j =
     1}{\overset{| \Gamma_0 | - 1}{\sum}} \left( \overset{d}{\underset{k =
     1}{\sum}} \frac{1}{2} M_{\Gamma} (\til{\xi}^2)^{- 1}_{i j}
     \partial_{w^i_k} \partial_{\til{w}_k^j} + \overset{d'}{\underset{k =
     1}{\sum}} M_{\Gamma} (\til{\xi}^2)^{- 1}_{i j} \partial_{q^i_k}
     \partial_{q_k^j} \right) \right)^{l_0} \]
  is $2 l_0$, to kill the variables $\{ \wbar_i \}_{i \in
  \Gamma_0}$ in $P_{j, k}$, the order of each term with respect to $\{
  w_i, \wbar_i \}_{i \in \Gamma_0}$ is at least
  \[ 2 \left( \underset{e \in \Gamma_1, 1 \leqslant i \leqslant
     d}{\overset{}{\sum}} n_{i, e} + d (| \Gamma_1 | - | \Gamma_0 | + 1)
     \right) . \]
  So the order of each term with respect to $\{ \vec{q}_i \}_{i \in \Gamma_0}$
  is at most
  \[ d' (| \Gamma_1 | - | \Gamma_0 | + 1) . \]
  Note the degree of $P_{j, k}$ with respect to variables $\{ \vec{q}_i \}_{i
  \in \Gamma_0}$ is also
  \[ d' (| \Gamma_1 | - | \Gamma_0 | + 1), \]
  there's no additional orders left to act on $\Phi$.
\end{proof}

Now, we consider the general case. Given a subgraph $\Gamma' \subseteq
\Gamma$, we use $\Gamma \backslash \Gamma'$ to denote the subgraph with the
vertices $\Gamma_0$, and the edges $\Gamma_1 \backslash
\Gamma'_1$.\footnote{Note $\Gamma \backslash \Gamma'$ is different from
$\Gamma / \Gamma'$. The latter is the graph obtained by contracting $\Gamma'$ to a single vertex.} We have the following result:

\begin{prop}\label{differential of Feynman graph}
  Given a connected decorated graph $(\Gamma, n)$ without self-loops, \ $\Phi
  \in \Omega^{\bu}_c ((\R^{d'} \times \C^d)^{| \Gamma_0 |})$,
  and a subgraph $\Gamma'$, we have
  \begin{enumerate}
    \item The integral
    \begin{eqnarray*}
      &  & \int_{\widetilde{(0, + \infty)^{| \Gamma_1' |}} /\R^+
      \times \widetilde{\left[ 0, \sqrt{L} \right]^{| \Gamma_1 \backslash
      \Gamma_1' |}}} t_{\tmop{square}}^{\ast} \int_{(\R^{d'} \times \C^d)^{| \Gamma_0 |}} \til{W} ((\Gamma, n), \Phi)\\
      & = & \sum_{n' \in \tmop{dec} (\Gamma / \Gamma')} C_{n'} W_0^L ((\Gamma
      / \Gamma', n'), D_{n'} \Phi),
    \end{eqnarray*}
    where $\tmop{dec} (\Gamma / \Gamma')$ is the set of all decorations of
    $\Gamma / \Gamma'$, and $C_{n'} = 0$ for all but finite many $n' \in
    \tmop{dec} (\Gamma / \Gamma')$, $D_{n'}$ are differential operators with constant coefficients, which only involve holomorphic derivatives.
    \item If $O_{(\Gamma', n |_{\Gamma'} \nobracket)} (\til{\Phi}) = 0$ for
    any $\til{\Phi} \in \Omega^{\bu}_c ((\R^{d'} \times \C^d)^{| \Gamma_0 |})$,
    \[ \int_{\widetilde{(0, + \infty)^{| \Gamma_1' |}} /\R^+ \times
       \widetilde{\left[ 0, \sqrt{L} \right]^{| \Gamma_1 \backslash \Gamma_1'
       |}}} t_{\tmop{square}}^{\ast} \int_{(\R^{d'} \times \C^d)^{| \Gamma_0 |}} \til{W} ((\Gamma, n), \Phi) = 0. \]
    \item When $L \rightarrow + \infty$,
    \[ \lim_{L \rightarrow + \infty} \int_{\widetilde{(0, + \infty)^{|
       \Gamma_1' |}} /\R^+ \times \widetilde{\left[ 0, \sqrt{L}
       \right]^{| \Gamma_1 \backslash \Gamma_1' |}}} t_{\tmop{square}}^{\ast}
       \int_{(\R^{d'} \times \C^d)^{| \Gamma_0 |}} \til{W}
       ((\Gamma, n), -) \]
    exists as a distribution.
  \end{enumerate}
\end{prop}

\begin{proof}
  From Proposition \ref{laman anomaly}, we have
  \begin{eqnarray*}
    &  & \int_{\widetilde{(0, + \infty)^{| \Gamma_1' |}} /\R^+ \times
    \widetilde{\left[ 0, \sqrt{L} \right]^{| \Gamma_1 \backslash \Gamma_1'
    |}}} t_{\tmop{square}}^{\ast} \int_{(\R^{d'} \times \C^d)^{| \Gamma_0 |}} \til{W} ((\Gamma, n), \Phi)\\
    & = & \pm \int_{\widetilde{(0, + \infty)^{| \Gamma_1' |}} /\R^+
    \times \widetilde{\left[ 0, \sqrt{L} \right]^{| \Gamma_1 \backslash
    \Gamma_1' |}}} t_{\tmop{square}}^{\ast} \int_{(\R^{d'} \times \C^d)^{| \Gamma_0 |}}\\
    &  & \til{W} ((\Gamma', n |_{\Gamma'} \nobracket), \til{W} ((\Gamma
    \backslash \Gamma', n |_{\Gamma \backslash \Gamma'} \nobracket), \Phi))\\
    & = & \pm \int_{\widetilde{\left[ 0, \sqrt{L} \right]^{| \Gamma_1
    \backslash \Gamma_1' |}}} \int_{(\R^{d'} \times \C^d)^{|
    (\Gamma / \Gamma')_0 | - 1}} \int_{\widetilde{(0, + \infty)^{| \Gamma_1'
    |}} /\R^+} \int_{(\R^{d'} \times \C^d)^{|
    \Gamma'_0 |}} t_{\tmop{square}}^{\ast}\\
    &  & \til{W} ((\Gamma', n |_{\Gamma'} \nobracket), \til{W} ((\Gamma
    \backslash \Gamma', n |_{\Gamma \backslash \Gamma'} \nobracket), \Phi))\\
    & = & \pm \int_{\widetilde{\left[ 0, \sqrt{L} \right]^{| \Gamma_1
    \backslash \Gamma_1' |}}} \int_{(\R^{d'} \times \C^d)^{|
    (\Gamma / \Gamma')_0 | - 1}}\\
    &  & O_{(\Gamma', n |_{\Gamma'} \nobracket)} (\til{W} ((\Gamma
    \backslash \Gamma', n |_{\Gamma \backslash \Gamma'} \nobracket), \Phi))\\
    & = & \pm \int_{\widetilde{\left[ 0, \sqrt{L} \right]^{| \Gamma_1
    \backslash \Gamma_1' |}}} \int_{(\R^{d'} \times \C^d)^{|
    (\Gamma / \Gamma')_0 | }} t_{\tmop{square}}^{\ast}\\
    &  & (D \til{W} ((\Gamma \backslash \Gamma', n |_{\Gamma \backslash
    \Gamma'} \nobracket), \Phi)) \left|_{\overline{{\triangle}}_{\Gamma'}} \right.,
  \end{eqnarray*}
  Here, \[\overline{{\triangle}}_{\Gamma'}=\{\left. \left( (z^1, x^1),
    (z^2, x^2), \ldots, \left(
    z^{| \Gamma_0 |}, x^{| \Gamma_0 |}
    \right) \right)  \in(\R^{d'} \times \C^d)^{|\Gamma_{0}|}\right|(z^i, x^i)=(z^j, x^j)\text{ if }i,j\in \Gamma'_{0}\},\]
  and $D$ is a differential operator with constant coefficients, which only
  involves holomorphic derivatives. 
  
  By Leibniz' rule, it is easy to
  see that $(D \til{W} ((\Gamma \backslash \Gamma', n |_{\Gamma \backslash
    \Gamma'} \nobracket), \Phi))$ is a linear combination of products of derivatives of propagators (in Schwinger space) and compactly supported smooth differential forms. The restriction of $(D \til{W} ((\Gamma \backslash \Gamma', n |_{\Gamma \backslash
    \Gamma'} \nobracket), \Phi))$ to $\overline{{\triangle}}_{\Gamma'}$ allows us to identify all the vertices in $\Gamma'$ with a single vertex in the quotient graph $\Gamma / \Gamma'$. As a consequence, we have
  \begin{eqnarray*}
    &  & \int_{\widetilde{(0, + \infty)^{| \Gamma_1' |}} /\R^+ \times
    \widetilde{\left[ 0, \sqrt{L} \right]^{| \Gamma_1 \backslash \Gamma_1'
    |}}} t_{\tmop{square}}^{\ast} \int_{(\R^{d'} \times \C^d)^{| \Gamma_0 |}} \til{W} ((\Gamma, n), \Phi)\\
    & = & \sum_{n' \in \tmop{dec} (\Gamma / \Gamma')} C_{n'} W_0^L ((\Gamma /
    \Gamma', n'), D_{n'} \Phi),
  \end{eqnarray*}
where $C_{n'} = 0$ for all but finitely many $n' \in \tmop{dec} (\Gamma /
  \Gamma')$, $D_{n'}$ are differential operators with constant coefficients, which only involve holomorphic derivatives.
  
  Other statements follow easily.
\end{proof}

\subsubsection{Anomaly vanishing results}

From last subsection, we know the anomalies are from local holomorphic
functionals $O_{(\Gamma, n)}$. In this section, we prove two vanishing results
for $O_{(\Gamma, n)}$.

\begin{prop}\label{anomaly free theorem2}
  Given a connected decorated Laman graph $(\Gamma, n)$ without self-loops,
  and $\Phi \in \Omega^{\bu}_c ((\R^{d'} \times \C^d)^{|
  \Gamma_0 |})$. If $d' \geqslant 1$ and the 1st Betti number $h_1 (\Gamma) =
  \dim (H_1 (\Gamma))$ is odd, we have
  \[ O_{(\Gamma, n)} (\Phi) = 0. \]
  \begin{proof}
    Let $H \in \R^{d'}$ be a hyperplane passing through the origin, we
    choose a unit normal vector of $H$ and denote it by $\vec{n}$. Then we
    define a map $r : (\R^{d'})^{| \Gamma_0 |} \rightarrow
    (\R^{d'})^{| \Gamma_0 |}$ by the following formula:
    \begin{eqnarray*}
      &  & r \left( q^1, q^2, \ldots,
      q^{| \Gamma_0 | - 1}, q^{| \Gamma_0 |}
      \right)\\
      & = & \left( q^1 - 2 (q^1 \cdot
      \vec{n}) \vec{n}, q^2 - 2 (q^2 \cdot
      \vec{n}) \vec{n}, \ldots, q^{| \Gamma_0 | - 1} - 2
      \left( q^{| \Gamma_0 | - 1} \cdot \vec{n} \right)
      \vec{n}, q^{| \Gamma_0 |} \right),
    \end{eqnarray*}
    i.e., $r$ is the reflection with respect to a hyperplane through
    $x^{| \Gamma_0 |} .$ From $\left( \ref{integrand}
    \right)$ we can get
    \[ r^{\ast} \til{W} ((\Gamma, n), \Phi) = (- 1)^{| \Gamma_1 |} \til{W}
       ((\Gamma, n), r^{\ast} \Phi) . \]
    Note $r$ will change the orientation of $(\R^{d'} \times \C^d)^{| \Gamma_0 |}$ by a $(- 1)^{1 - | \Gamma_0 |}$ factor,
    we have
    \begin{eqnarray*}
      O_{(\Gamma, n)} (\Phi) & = & \int_{\partial C_{\Gamma_1}}
      t_{\tmop{square}}^{\ast} \int_{(\R^{d'} \times \C^d)^{|
      \Gamma_0 |}} \til{W} ((\Gamma, n), \Phi)\\
      & = & (- 1)^{| \Gamma_1 | - | \Gamma_0 | + 1} \int_{\partial
      C_{\Gamma_1}} t_{\tmop{square}}^{\ast} \int_{(\R^{d'} \times \C^d)^{| \Gamma_0 |}} \til{W} ((\Gamma, n), r^{\ast}
      \Phi)\\
      & = & (- 1)^{h_1 (\Gamma)} O_{(\Gamma, n)} (r^{\ast} \Phi)\\
      & = & (- 1)^{h_1 (\Gamma)} O_{(\Gamma, n)} (\Phi)
    \end{eqnarray*}
    In the last step, we used Proposition \ref{laman anomaly}: since
    $O_{(\Gamma, n)}$ is a local holomorphic functional, it only depends on
    the values of $\Phi$ and holomorphic derivatives on the diagonal. However,
    $\Phi$ and $r^{\ast} \Phi$ have the same values and holomorphic
    derivatives on the diagonal.
    
    Now, it is obvious that $O_{(\Gamma, n)} (\Phi) = 0$ if $h_1 (\Gamma)$ is
    odd.
  \end{proof}
\end{prop}

Before state the second vanishing result, we prove the following lemma:

\begin{lem}
  \label{Kontsevich's lemma}Let $d = 0$ and $d' = 2$. Given a connected
  decorated graph $(\Gamma, n)$ without self-loops, such that $h_1 (\Gamma) =
  | \Gamma_1 | - | \Gamma_0 | + 1 \geqslant 1$. We have the following
  equality:
  \[ \int_{(\R^2)^{| \Gamma_0 | - 1}} \tilde{W} ((\Gamma, n), 1) = 0.
  \]
\end{lem}
\begin{proof}
  Note $\R^2 \cong \C$ as real vector spaces, we can use
  \[ z = x_1 + i x_2 \]
  to rephrase the propagator:
  \begin{eqnarray*}
    P_t (\vec{x}) & = & \frac{1}{\pi} e^{- v \cdot v} d^2 v\\
    & = & \frac{i}{8 \pi} e^{- \frac{z \bar{z}}{4 t}} d \left(
    \frac{z}{\sqrt{t}} \right) d \left( \frac{\bar{z}}{\sqrt{t}} \right)\\
    & = & \frac{i}{8 \pi} e^{- \frac{z \bar{z}}{4 t}} \left( \frac{d (z) d
    (\bar{z})}{t} + \frac{z d (\bar{z}) d t}{2 t^2} - \frac{\overline{z} d (z)
    d t}{2 t^2} \right)\\
    & = & P_t (z, \bar{z})
  \end{eqnarray*}

  To simplify expressions, we introduce
  \[ \left\{\begin{array}{l}
       P_t^1 (z, \bar{z}) = \frac{i}{8 \pi} e^{- \frac{z \bar{z}}{4 t}}
       \frac{d (z) d (\bar{z})}{t}\\
       P_t^2 (z, \bar{z}) = \frac{i}{8 \pi} e^{- \frac{z \bar{z}}{4 t}}
       \frac{z d (\bar{z}) d t}{2 t^2}\\
       P_t^3 (z, \bar{z}) = \frac{i}{8 \pi} e^{- \frac{z \bar{z}}{4 t}}
       \frac{\overline{z} d (z) d t}{2 t^2}
     \end{array}\right. \]
  In the expansion of
  \[ \tilde{W} ((\Gamma, n), 1) = \prod_{e \in \Gamma_1 } P_{t_e} (z^{h (e)} -
     z^{t (e)}, \bar{z}^{h (e)} - \bar{z}^{t (e)}), \]
  each monomial has the following form:
  \[ \prod_{e \in S_1} P_{t_e}^1 (z^e, \bar{z}^e) \prod_{e \in S_2} P_{t_e}^2
     (z^e, \bar{z}^e) \prod_{e \in S_3} (- P_{t_e}^3 (z^e, \bar{z}^e)), \]
  where
  \[ \left\{\begin{array}{l}
       z^e = z^{h (e)} - z^{t (e)}\\
       \bar{z}^e = \bar{z}^{h (e)} - \bar{z}^{t (e)}
     \end{array}\right., \]
  $S_1, S_2, S_3 \subseteq \Gamma_1$ are disjoint subsets, such that
  \[ S_1 \cup S_2 \cup S_3 = \Gamma_1 . \]
  The integration of such monomials over $(\R^2)^{| \Gamma_0 | - 1}
  \cong (\C)^{| \Gamma_0 | - 1}$ are nonzero only if
  \[ | S_2 | = | S_3 | = | \Gamma_1 | - | \Gamma_0 | + 1 = h_1 (\Gamma) , \]
  because these monomials should be top forms.
  As a consequence, we have
  \begin{eqnarray*}
    &  & \int_{(\C)^{| \Gamma_0 | - 1}} \prod_{e \in S_1} P_{t_e}^1
    (z^e, \bar{z}^e) \prod_{e \in S_2} P_{t_e}^2 (z^e, \bar{z}^e) \prod_{e \in
    S_3} (- P_{t_e}^3 (z^e, \bar{z}^e))\\
    & = & (- 1)^{h_1 (\Gamma)} \int_{(\C)^{| \Gamma_0 | - 1}}
    \prod_{e \in S_1} P_{t_e}^1 (z^e, \bar{z}^e) \prod_{e \in S_2} P_{t_e}^2
    (z^e, \bar{z}^e) \prod_{e \in S_3} P_{t_e}^3 (z^e, \bar{z}^e) .
  \end{eqnarray*}
  If we sum over all possible $S_1, S_2, S_3 \subseteq \Gamma_1$, we get
  \begin{eqnarray*}
    &  & \int_{(\C)^{| \Gamma_0 | - 1}}\tilde{W} ((\Gamma, n), 1)\\
    & = & (- 1)^{h_1 (\Gamma)} \int_{(\C)^{| \Gamma_0 | - 1}}
    \prod_{e \in \Gamma_1 } (P_{t_e}^1 (z^e, \bar{z}^e) + P_{t_e}^2 (z^e,
    \bar{z}^e) + P_{t_e}^3 (z^e, \bar{z}^e)) .
  \end{eqnarray*}
  We note
  \begin{eqnarray*}
    P_t^2 (z, \bar{z}) + P_t^3 (z, \bar{z}) & = & \frac{i}{8 \pi} e^{- \frac{z
    \bar{z}}{4 t}} \frac{z d (\bar{z}) d t}{2 t^2} + \frac{i}{8 \pi} e^{-
    \frac{z \bar{z}}{4 t}} \frac{\overline{z} d (z) d t}{2 t^2}\\
    & = & d_{\tmop{deRham}} \left( \frac{- i}{4 \pi t} e^{- \frac{z
    \bar{z}}{4 t}} d t \right),
  \end{eqnarray*}
  where $d_{\tmop{deRham}}$ is the de Rham differential for variables of
  spacetime $\{ z, \bar{z} \}$.
  
  In the expansion of
  \[ \prod_{e \in \Gamma_1 } (P_{t_e}^1 (z^e, \bar{z}^e) + P_{t_e}^2 (z^e,
     \bar{z}^e) + P_{t_e}^3 (z^e, \bar{z}^e)), \]
  each monomial has the following form:
  \[ \prod_{e \in S_1} P_{t_e}^1 (z^e, \bar{z}^e) \prod_{e \in S_2} (P_{t_e}^2
     (z^e, \bar{z}^e) + P_{t_e}^3 (z^e, \bar{z}^e)), \]
  where $S_1, S_2 \subseteq \Gamma_1$ are disjoint subsets, such that $S_1
  \cup S_2 = \Gamma_1$. The integration of such monomials over
  $(\C)^{| \Gamma_0 | - 1}$ are nonzero only if
  \[ | S_2 | = 2 h_1 (\Gamma) . \]
  Note $h_1 (\Gamma) \geqslant 1$ and
  \[ d_{\tmop{deRham}} P_t^1 (z, \bar{z}) = 0, \]
  each monomial which contributes to the integral is $d_{\tmop{deRham}}$
  exact. Furthermore, they decrease exponentially at infinity. By Stokes
  theorem, we conclude that
  \[ \int_{(\R^2)^{| \Gamma_0 | - 1}} \tilde{W} ((\Gamma, n), 1) = 0.
  \]
\end{proof}

\begin{rmk}
  The above lemma is a Schwinger space version of the key lemma used in the
  proof of Kontsevich's formality theorem, see {\cite[Lemma
  6.4]{kontsevich2003deformation}}.
\end{rmk}

Now, we can prove the second vanishing result:

\begin{prop}\label{vanishing of anomaly121}
  Given a connected decorated Laman graph $(\Gamma, n)$ without self-loops,
  and $\Phi \in \Omega^{\bu}_c ((\R^{d'} \times \C^d)^{|
  \Gamma_0 |})$. If $d' \geqslant 2$ and $| \Gamma_0 | \geqslant 3$, we have
  \[ O_{(\Gamma, n)} (\Phi) = 0. \]
\end{prop}

\begin{proof}
  First, it is easy to prove that $h_1 (\Gamma) \geqslant 1$ by the Laman conditions.  

    Let $f$ be the map defined by
    \begin{align*}
       f:(\R^{d'} \times \C^d)^{| \Gamma_0 |} & \rightarrow (\R^{d'} \times \C^d)^{| \Gamma_0 |} \\
       \left( (z^1, x^1),
    (z^2, x^2), \ldots, \left(
    z^{| \Gamma_0 |}, x^{| \Gamma_0 |}
    \right) \right) & \rightarrow \left( (z^1, x^{| \Gamma_0 |}),
    (z^2, x^{| \Gamma_0 |}), \ldots, \left(
    z^{| \Gamma_0 |}, x^{| \Gamma_0 |}
    \right) \right).
    \end{align*}
    By Proposition
  \ref{laman anomaly}, we have
  \[
  O_{(\Gamma, n)} (\Phi)=\int_{(\R^{d'} \times \C^d)_{\left(
     w^{| \Gamma_0 |}, q^{| \Gamma_0 |}
      \right)}} (D \Phi) \left|_{w^i = 0,
      q^i = 0, i \neq | \Gamma_0 |} \right.,
  \]
  where $D$ is a differential operator with constant coefficients, which
    only involves holomorphic derivatives. In particular, $D$ does not contain derivatives of $\{q^i\}_{i\in\Gamma_{0}}$. As a consequence,
  \[ O_{(\Gamma, n)} (\Phi) = O_{(\Gamma, n)} \left( f^{*}\Phi \right),
  \]
  we only need to prove
  \[ O_{(\Gamma, n)} \left( f^{*}\Phi  \right) = 0. \]
  Recall
  \begin{eqnarray*}
    &  & O_{(\Gamma, n)} \left( f^{*}\Phi\right)\\
    & = & \pm \int_{\partial C_{\Gamma_1}} t_{\tmop{square}}^{\ast}
    \int_{(\R^{d'} \times \C^d)_{\left(
     w^{| \Gamma_0 |}, q^{| \Gamma_0 |}
      \right)}}\int_{(\C^d \times \R^{d'})^{| \Gamma_0 |-1}} \tilde{W}
    \left( (\Gamma, n), f^{*}\Phi \right)\\
    & = & \pm \int_{\partial C_{\Gamma_1}} t_{\tmop{square}}^{\ast}
    \int_{(\R^{d'} \times \C^d)_{\left(
     w^{| \Gamma_0 |}, q^{| \Gamma_0 |}
      \right)}}\int_{(\C^d \times \R^{d' - 2})^{| \Gamma_0 |-1}}
    \int_{(\R^2)^{| \Gamma_0 |-1}} \tilde{W} \left( (\Gamma, n), f^{*}\Phi \right).
  \end{eqnarray*}
  We notice that the propagator in Schwinger space can be factorized as follows:
  \begin{eqnarray*}
     & &P_t (z - w, \zbar -
    \wbar, x-y)\\
    &=& \frac{1}{\pi^{d + \frac{d'}{2}}} e^{- (z - w) \cdot u -
    v \cdot v} \d^d u \d^{d'} v\\
    &=&\left(\frac{1}{\pi^{d + \frac{d'-2}{2}}} e^{- (z - w) \cdot u -\sum_{i=1}^{d'-2}
    v_{i} v_{i}} \d^d u \prod_{i=1}^{d'-2}\d v_{i}\right)\left(\frac{1}{\pi}e^{ -v_{d'-1} v_{d'-1}-v_{d'} v_{d'}}\d v_{d'-1}\d v_{d'}\right).
  \end{eqnarray*}
  Moreover, $f^{*}\Phi$ is a constant with respect to the coordinates $\{q^i\}_{i\in\Gamma_0,i\neq|\Gamma_{0}|}$, so by Lemma
  \ref{Kontsevich's lemma}, we have
  \[ O_{(\Gamma, n)} (\Phi) = 0. \]
\end{proof}

Finally, we can prove our main theorem on the existence of quantization for topological-holomorphic theories:
\begin{thm}\label{anomaly free theorem1}Given the space of fields $\cE$ defined in $(\ref{space of fields})$, and an interaction $I$, which satisfies the assumption $(\ref{weaker assumption})$ and the classical master equation $(\ref{CME})$.
If $d'>1$, then the quantum master equation is satisfied. In particular, the observables of the resulting quantum topological-holomorphic theory defines a factorization algebra on $\R^{d'} \times \C^d$.
\end{thm}
\begin{proof}
By Corollary \ref{weight QME}, we only need to show the following equality for stable graph $\gamma$: 
\begin{eqnarray}
& & \lim_{\epsilon\rightarrow 0}\sum_{e \in E (\gamma)}  w_{\gamma, e} (P_{\epsilon < L}, Q (P_{\epsilon < L}),
     I) + \lim_{\epsilon\rightarrow 0}\sum_{e \in E (\gamma)} w_{\gamma, e} (P_{\epsilon <
     L}, K_L, I) \label{QME2}\\
     & - & \lim_{\epsilon\rightarrow 0}\sum_{e \in E (\gamma)}\lim_{t \rightarrow 0} w_{\gamma, e} (P_{\epsilon < L}, K_t, I)  \nonumber \\
     & = & 0. \nonumber
\end{eqnarray}
By equation \eqref{QME}, we have
\begin{eqnarray}
& &   W_0^L ((\Gamma, n), (\bar{\partial} + \d_{deRham})\Phi) \pm   \int_{\partial_{\sqrt{L}}
    \widetilde{\left[ 0, \sqrt{L} \right]^{| \Gamma_1 |}}} t_{\tmop{square}}^{\ast}\int_{(\R^{d'} \times \C^d)^{| \Gamma_0 |}} \til{W} ((\Gamma, n), \Phi)\label{QME3}\\
    & \pm &  \int_{\partial_{0}
    \widetilde{\left[ 0, \sqrt{L} \right]^{| \Gamma_1 |}}}t_{\tmop{square}}^{\ast} \int_{(\R^{d'} \times \C^d)^{| \Gamma_0 |}} \til{W} ((\Gamma, n), \Phi)\nonumber\\
    & = & 0.\nonumber
\end{eqnarray}
We will show that $(\ref{QME2})$ and $(\ref{QME3})$ are equivalent.

First, we can see that the first term in $(\ref{QME2})$ corresponds to the first term in $(\ref{QME3})$ (after integration by parts). 

Recall that
\[\partial_{\sqrt{L}} \widetilde{\left[ 0, \sqrt{L} \right]^{| \Gamma_1 |}}
     = \bigcup_{e \in \Gamma_1} (- 1)^{| e |} \left\{ \sqrt{L} \right\} \times
     \widetilde{\left[ 0, \sqrt{L} \right]^{| (\Gamma / e)_1 |}}.
\] For a given edge $e\in\Gamma_{1}$, we consider the integral
\[\int_{\left\{ \sqrt{L} \right\} \times
     \widetilde{\left[ 0, \sqrt{L} \right]^{| (\Gamma / e)_1 |}}}
t_{\tmop{square}}^{\ast}  \int_{(\R^{d'} \times \C^d)^{| \Gamma_0 |}} \til{W} ((\Gamma, n), \Phi).
\]  We notice that in the integrand:
    \begin{enumerate}
        \item The propagator (in Schwinger space) assigned to $e$ becomes \[H(t, z^{h(e)} - z^{t(e)}, \zbar^{h(e)} -
    \zbar^{t(e)}, x^{h(e)}-x^{t(e)}),\]because the first term in $(\ref{schwinger propagator})$ vanished by degree reason.
    \item The propagator (in schwinger space) assigned to $e'\neq e$ becomes \[
    - \d t_{e'} \wedge (\bar{\partial}^{\ast}_{z_{h(e')}} + \d_{deRham,x_{h(e')}}^{\ast}) H(t, z^{h(e')} - z^{t(e')}, \zbar^{h(e')} -
    \zbar^{t(e')}, x^{h(e')}-x^{t(e')}),
    \]
    because the second term in $(\ref{schwinger propagator})$ is not a top form.
    \end{enumerate}
    Therefore, the second term in $(\ref{QME2})$ corresponds to the second term in $(\ref{QME3})$.

Finally, by Proposition \ref{vanishing of anomaly121} and Proposition \ref{differential of Feynman graph} $(2)$, the third term in $(\ref{QME3})$ only have contributions from the integrals over 
\[
     \bigcup_{e\in \Gamma_{1}} (- 1)^{|e|} \{0\} \times
     \widetilde{\left[ 0, + \sqrt{L} \right]^{| (\Gamma / e)_1
     |}},
\] because the only Laman gragh with $2$ vertices is the interval. So by similar arguments, it corresponds to the third term in $(\ref{QME2})$. This completes the proof of quantum master equation. The existence of factorization algebra follows from \cite{CG2}.
\end{proof}

As a corollary we obtain the following result about well-studied topological-holomorphic gauge theories.

\begin{cor}
Let $\lie{g}$ be a Lie algebra which is equipped with an invariant, non-degenerate, symmetric pairing.
Then, for $d' > 1$, hybrid Chern--Simons theory on $\R^{d'} \times \C^d$ with gauge Lie algebra $\lie{g}$ admits a quantization, hence factorization algebra of observables, to all orders in perturbation theory.

In particular, consider the list of holomorphic topological theories given at the end of section \ref{s:hcs}.
The theories (1),(2),(6),(7),(8),(9) admit quantizations to all orders in perturbation theory for any gauge Lie algebra 
\end{cor}

\begin{rmk}
The existence of the quantization of ordinary Chern--Simons theory is well-known.
The existence of the quantization of (2) has already been proven in the work that the theory was introduced \cite{Yangian}.
\end{rmk}

\appendix

\section{Graph theory background}\label{graph theory}

In this appendix, we introduce some concepts and facts from graph theory.
These facts will be used in the study of Feynman graph integrals.

\begin{dfn}
  A {\tmstrong{directed graph}} $\Gamma$ consists of the following data:
  \begin{enumeratenumeric}
    \item A set of vertices $\Gamma_0$ and a set of edges $\Gamma_1$.
    
    \item An ordering of all the vertices $\Gamma_0$ and an ordering of all
    the edges $\Gamma_1$.
    
    \item Two maps
    \[ t, h : \Gamma_1 \rightarrow \Gamma_0 \]
    which are the assignments of tail and head to each directed edge.
    
    Furthermore, we say $\Gamma$ is {\tmstrong{decorated}} in $\C^d
    \times \R^{d'}$ if we have a map
    \[ n : e \in \Gamma_1 \rightarrow (n_{1, e}, n_{2, e}, \ldots, n_{d, e})
       \in (\mathbb{Z}^{\geqslant 0})^d, \]
    which associates each edge $d$ non-negative integer. We use $(\Gamma, n)$
    to denote a decorated graph. If $n$ is the zero map, we will simply write
    $\Gamma$ for $(\Gamma, n)$. We say $\Gamma$ is a {\tmstrong{graph without
    self-loops}}, if $t (e) \neq h (e)$ for any $e \in \Gamma_1$.
  \end{enumeratenumeric}
\end{dfn}

\

We will use $| \Gamma_0 |$ and $| \Gamma_1 |$ to denote the number of
vertices and edges respectively.

\begin{dfn}
  Given a connected directed graph $\Gamma$ without self-loops, we have the
  following {\tmstrong{incidence matrix}}:
  \[ \rho^e_i = \left\{\begin{array}{l}
       1 \text{\qquad if } h (e) = i\\
       - 1 \text{\quad if } t (e) = i\\
       0 \text{\qquad otherwise}
     \end{array}\right. \]
  where $i \in \Gamma_0$, $e \in \Gamma_1$.
\end{dfn}

\begin{dfn}
  Given a connected directed graph $\Gamma$, a tree $T \subseteq \Gamma$ is
  said to be a spanning tree if every vertex of $\Gamma$ lies in $T$. We
  denote the set of all spanning tree by $\tmop{Tree} (\Gamma)$.
\end{dfn}

\begin{dfn}
  Given a connected directed graph $\Gamma$ without self-loops and two
  disjoint subsets of vertices $V_1, \text{ } V_2 \subseteq \Gamma_0$, we
  define $\tmop{Cut} (\Gamma ; V_1, V_2)$ to be the set of subsets $C
  \subseteq \Gamma_1$ satisfying the following properties:
  \begin{enumeratenumeric}
    \item The removing of edges in $C$ from $\Gamma$ divides $\Gamma$ into
    exactly two connected trees, which we denoted by $\Gamma' (C), \text{ }
    \Gamma'' (C) $, such that $V_1 \subseteq \Gamma_0' (C), \text{ } V_2
    \subseteq \Gamma_0'' (C)$.
    
    \item $C$ doesn't contain any proper subset satisfying property 1.
  \end{enumeratenumeric}
\end{dfn}

\begin{dfn}
  Given a connected directed graph $\Gamma$ without self-loops, and a function
  maps each $e \in \Gamma_1$ to $t_e \in (0, + \infty)$, we define the
  {\tmstrong{weighted laplacian}} of $\Gamma$ by the following formula:
  \[ M_{\Gamma} (t)_{ij} = \sum_{e \in \Gamma_1} \rho^e_i \frac{1}{t_e}
     \rho_j^e, \text{\quad} 1 \leqslant i, j \leqslant | \Gamma_0 | - 1. \]
\end{dfn}

The following facts will be used to show the finitenes of Feynman graph
integrals.

\begin{thm}[Kirchhoff]
  Given a connected graph $\Gamma$ without self-loops, the determinant of
  weighted laplacian is given by the following formula:
  \[ \det M_{\Gamma} (t) = \frac{\underset{T \in \tmop{Tree} (\Gamma)}{\sum}
     \underset{e \nin T}{\prod} t_e}{\underset{e \in \Gamma_1}{\prod} t_e} \]
\end{thm}

\begin{cor}
  The inverse of $M_{\Gamma} (t)$ is given by the following
  formula:\label{Minverse}
  \[ (M_{\Gamma} (t)^{- 1})^{i j} = \frac{1}{\underset{T \in \tmop{Tree}
     (\Gamma)}{\sum} \underset{e \nin T}{\prod} t_e} \cdot \left( \sum_{C \in
     \tmop{Cut} (\Gamma ; \{ i, j \}, \{ | \Gamma_0 | \})} \prod_{e \in C} t_e
     \right) \]
\end{cor}

\begin{cor}
  \label{boundness}We have the following equality:
  \begin{eqnarray*}
    &  & \frac{1}{t_e} \sum_{i = 1}^{| \Gamma_0 | - 1} \rho^e_i (M_{\Gamma}
    (t)^{- 1})^{i j}\\
    & = & \frac{1}{\underset{T \in \tmop{Tree} (\Gamma)}{\sum} \underset{e
    \nin T}{\prod} t_e} \left( \sum_{C \in \tmop{Cut} (\Gamma ; \{ j, h (e)
    \}, \{ | \Gamma_0 |, t (e) \})} \frac{\prod_{e' \in C} t_{e'}}{t_e}
    \right.\\
    & - & \left. \sum_{C \in \tmop{Cut} (\Gamma ; \{ j, t (e) \}, \{ |
    \Gamma_0 |, h (e) \})} \frac{\prod_{e' \in C} t_{e'}}{t_e} \right)
  \end{eqnarray*}
  In particular, every term of the numerator also appears in the denominator,
  so we have
  \[ \left| \frac{1}{t_e} \sum_{i = 1}^{| \Gamma_0 | - 1} \rho^e_i (M_{\Gamma}
     (t)^{- 1})^{i j} \right| \leqslant 2. \]
  We use $(d_{\Gamma}  (t)^{- 1})^{e j}$ to denote $\frac{1}{t_e} \sum_{i =
  1}^{| \Gamma_0 | - 1} \rho^e_i (M_{\Gamma} (t)^{- 1})^{i j}$.
\end{cor}

\begin{proof}
  See {\cite[Appendix B]{Li:2011mi}}.
\end{proof}

\begin{rmk}
  If we view graphs as discrete spaces, the incidence matrix can be viewed as
  de Rham differential. Then $(d_{\Gamma}  (t)^{- 1})^{e j}$ can be viewed as
  the Green's function of de Rham differential on a graph. This might explain
  the importance of $(d_{\Gamma}  (t)^{- 1})^{e j}$.
\end{rmk}

Finally, we introduce the concept of Laman graphs, which will be useful to
describe anomalies of Feynman graph integrals.

\begin{dfn}
  \label{Laman graph1}Given a connected directed graph $\Gamma$ without
  self-loops, we call $\Gamma$ a Laman graph for $\C^d \times
  \R^{d'}$, if the following conditions hold:
  \begin{enumeratenumeric}
    \item For any subgraph $\Gamma'$ such that $| \Gamma'_0 | \geqslant 2$, we
    have the following inequality:
    \begin{equation}
      (d + d') | \Gamma'_0 | \geqslant (d + d' - 1) | \Gamma_1' | + d + d' +
      1. \label{Laman condition1}
    \end{equation}
    \item The following equality holds:
    \begin{equation}
      (d + d') | \Gamma_0 | = (d + d' - 1) | \Gamma_1 | + d + d' + 1.
      \label{Laman condition2}
    \end{equation}
  \end{enumeratenumeric}
  We use the notation $\tmop{Laman} (\Gamma)$ to denote all the Laman
  subgraphs of $\Gamma$.
\end{dfn}

\begin{rmk}
  When $d + d' = 2$, the concept of Laman graph originate from Laman's
  characterization of generic rigidity of graphs embeded in two dimensional
  real linear space(see {\cite{Laman1970OnGA}}). Our definition appears in
  {\cite{budzik2023feynman}}.
\end{rmk}

\section{Compactification of Schwinger spaces}\label{Schwinger spaces}

In this appendix, we review the construction of a compactification of
Schwinger space in {\cite{wang2024feynman}}. Some constructions in this appendix is similar to the compactification of configuration spaces in \cite{axelrod1994chern,kontsevich1994feynman}, and can be viewed as an example of the constructions in \cite{Ammann2019ACO}.

\begin{dfn}
  Given a directed graph $\Gamma$, and $L > 0$, the {\tmstrong{Schwinger
  space}} is defined by $(0, L]^{| \Gamma_1 |}$. The orientation is given by
  the following formula:
  \[ \int_{(0, L]^{| \Gamma_1 |}} \prod_{e \in \Gamma_1} d t_e = L^{| \Gamma_1
     |} . \]
\end{dfn}

Assume $\Gamma$ is a connected directed graph, $L > 0$ is a positive real
number. Let $S \subseteq \Gamma_1$ be a subset of $\Gamma_1$, we define the
following submanifold with corners of $[0, L]^{| \Gamma_1 |}$:
\[ \Delta_S = \left\{ (t_1, t_2, \ldots, t_{| \Gamma_1 |}) \in [0, L]^{|
   \Gamma_1 |} | \nobracket t_e = 0 \quad \tmop{if} e \in S \right\} . \]
The compactification of Schwinger space is obtained by iterated real blow ups
of $[0, L]^{| \Gamma_1 |}$ along $\Delta_S$ for all $S \subseteq \Gamma_1$ in
a certain order (see {\cite{epub47792,Ammann2019ACO}}). To avoid getting into
technical details of real blow ups of manifolds with corners, we will use
another simpler definition. Instead, we present a typical example of real blow
up, which will be helpful to understand our construction:

\begin{eg}
  Let $S = \{ 1, 2, \ldots, k \} \subseteq \Gamma_1$, the real blow up of $[0,
  + \infty)^{| \Gamma_1 |}$ along $\Delta_S$ is the following manifold(with
  corners):
  \[ [[0, + \infty)^{| \Gamma_1 |} : \Delta_S] \assign \left\{ (\rho, \xi_1,
     \ldots, \xi_k, t_{k + 1}, \ldots t_{| \Gamma_1 |}) \in [0, + \infty)^{|
     \Gamma_1 | + 1} \left| \sum_{i = 1}^k \xi_i^2 = 1 \right. \right\} . \]
  We have a natural map from $[[0, + \infty)^{| \Gamma_1 |} : \Delta_S]$ to
  $[0, + \infty)^{| \Gamma_1 |}$:
  \[ (\rho, \xi_1, \ldots, \xi_k, t_{k + 1}, \ldots t_{| \Gamma_1 |})
     \rightarrow (t_1 = \rho \xi_1, \ldots, t_k = \rho \xi_k, t_{k + 1},
     \ldots, t_{| \Gamma_1 |}) . \]
  We also have a natural inclusion map from $(0, + \infty)^{| \Gamma_1 |}$ to
  $[[0, + \infty)^{| \Gamma_1 |} : \Delta_S]$:
  \begin{eqnarray*}
    &  & (t_1, \ldots, t_{| \Gamma_1 |})\\
    & \rightarrow & \left( \rho = \sqrt{\sum_{i = 1}^k t_i^2}, \xi_1 =
    \frac{t_1}{\sqrt{\sum_{i = 1}^k t_i^2}}, \ldots, \xi_k =
    \frac{t_k}{\sqrt{\sum_{i = 1}^k t_i^2}}, t_{k + 1}, \ldots, t_{| \Gamma_1
    |} \right) .
  \end{eqnarray*}
  For us, the most important property is that $\frac{t_i}{\sqrt{\sum_{i = 1}^k
  t_i^2}}$ can be extended to a smooth function $\xi_i$ on $[[0, + \infty)^{|
  \Gamma_1 |} : \Delta_S]$.
\end{eg}

These inclusion maps determine a map into the product over all $S \subset \Gamma_1$ which we denote
\[ i : (0, + \infty)^{| \Gamma_1 |} \rightarrow \prod_{S \subseteq \Gamma_1}
   [[0, + \infty)^{| \Gamma_1 |} : \Delta_S] \]
\begin{dfn}\label{compactified Schwinger space}
  We call the closure of the image of $(0, L]^{| \Gamma_1 |}$ under $i$ the
  {\tmstrong{compactified Schwinger space}} of $\Gamma$. We denote it by
  $\widetilde{[0, L]^{| \Gamma_1 |}}$.
\end{dfn}

\begin{rmk}
  Sometimes, by abusing of concept, we will also call the closure of the image
  of $(0, + \infty)^{| \Gamma_1 |}$ under $i$ the compactified Schwinger
  space, although it is not compact. We will use $\widetilde{[0, + \infty)^{|
  \Gamma_1 |}}$ to denote it.
\end{rmk}

\begin{prop}
  $\widetilde{[0, L]^{| \Gamma_1 |}}$ is a compact manifold with corners.
\end{prop}

\begin{proof}
  See {\cite{Ammann2019ACO}}.
\end{proof}

To obtain a more concrete description of $\widetilde{[0, L]^{| \Gamma_1 |}}$,
we introduce a useful concept called corners of compactified Schwinger spaces.

For $\Gamma_1 = S_0 \supseteq S_1 \supsetneq S_2 \supsetneq \cdots \supsetneq
S_m \supsetneq S_{m + 1} = \varnothing$, we define a submanifold
\[ C_{S_1, S_2, \ldots, S_m} \subseteq [0, + \infty)^m \times (0, + \infty)^{|
   S_0 | - | S_1 |} \times \prod_{i = 1}^m (0, + \infty)^{| S_i | - | S_{i +
   1} |}, \]
which is given by a set of equations:

if we use $\{ \rho_i \} (1 \leqslant i \leqslant m)$ for the coordinates of
the first product component $[0, + \infty)^m$, use $\{ t_e \} (e \in S_0
\backslash S_1)$ for the coordinates of the product component $(0, +
\infty)^{| S_0 | - | S_1 |}$, use $\{ \xi_e^i \} \left( 1 \leqslant i
\leqslant m, e \in S_i {\backslash S_{i + 1}}  \right)$ for the coordinates of
the product component $(0, + \infty)^{| S_{i - 1} | - | S_i |}$, then:
\[ \left\{\begin{array}{l}
     \underset{e \in S_m}{\sum} (\xi_e^m)^2 = 1\\
     \underset{e \in S_{m - 1} {\backslash S_m} }{\sum} (\xi_e^{m - 1})^2 +
     \underset{{e \in S_m} }{\sum} (\rho_m \xi_e^m)^2 = 1\\
     \underset{e \in S_{m - 2} {\backslash S_{m - 1}} }{\sum} (\xi_e^{m -
     2})^2 + \underset{e \in S_{m - 1} {\backslash S_m} }{\sum} (\rho_{m - 1}
     \xi_e^{m - 1})^2 + \underset{{e \in S_m} }{\sum} (\rho_{m - 1} \rho_m
     \xi_e^m)^2 = 1\\
     \ldots\\
     \underset{i = 1}{\overset{k}{\sum}} \left( \underset{e \in S_{m - k + i}
     {\backslash S_{m - k + i + 1}} }{\sum} \left( \left( \underset{j =
     1}{\overset{i - 1}{\prod}} \rho_{m - k + j} \right) \xi_e^{m - k + i}
     \right)^2 \right) = 1\\
     \ldots\\
     \underset{i = 1}{\overset{m}{\sum}} \left( \underset{e \in S_{m - k + i}
     {\backslash S_{m - k + i + 1}} }{\sum} \left( \left( \underset{j =
     1}{\overset{i - 1}{\prod}} \rho_{m - k + j} \right) \xi_e^{m - k + i}
     \right)^2 \right) = 1
   \end{array}\right. \]
We have a natural map from $C_{S_1, S_2, \ldots, S_m}$ to $[0, + \infty)^{|
\Gamma_1 |}$: Let $\{ \widetilde{t_e} \}_{e \in \Gamma_1}$ be the coordinates
of $[0, + \infty)^{| \Gamma_1 |}$, then under the natural map, we have:
\[ \left\{\begin{array}{l}
     \widetilde{t_e} = t_e \text{\quad for } e \in S_0 \backslash S_1,\\
     \widetilde{t_e} = \rho_1 \xi_e^1 \text{\quad for } e \in S_1 {\backslash
     S_2} ,\\
     \widetilde{t_e} = \rho_1 \rho_2 \xi_e^2 \text{\quad for } e \in S_2
     {\backslash S_3} ,\\
     \ldots\\
     \widetilde{t_e} = \underset{k = 1}{\overset{m}{\prod}} \rho_k \xi_e^m
     \text{\quad for } e \in S_m {\backslash S_{m + 1}}  .
   \end{array}\right. \]
\begin{prop}
  \label{covered by corners}The following statements are true:
  \begin{enumeratenumeric}
    \item There is a natural inclusion map from $C_{S_1, S_2, \ldots, S_m}$ to
    $\widetilde{[0, + \infty)^{| \Gamma_1 |}}$.
    
    \item $\widetilde{[0, + \infty)^{| \Gamma_1 |}}$ is covered by the union
    of all $C_{S_1, S_2, \ldots, S_m}$:
    \[ \widetilde{[0, + \infty)^{| \Gamma_1 |}} = \bigcup_{\Gamma_1 = S_0
       \supseteq S_1 \supsetneq \cdots \supsetneq S_m \supsetneq S_{m + 1} =
       \varnothing} C_{S_1, S_2, \ldots, S_m} . \]
    \item There is a natural action of $\R_+ = (0, + \infty)$ on
    $\widetilde{[0, + \infty)^{| \Gamma_1 |}}$, which has the following form
    on $(0, + \infty)^{| \Gamma_1 |} \subseteq \widetilde{[0, + \infty)^{|
    \Gamma_1 |}}$:
    \[ \lambda \cdot (t_e)_{e \in \Gamma_1} = (\lambda t_e)_{e \in \Gamma_1}
       \quad \lambda \in (0, + \infty) . \]
  \end{enumeratenumeric}
\end{prop}

\begin{proof}
  Left as an excercise for the reader.
\end{proof}

\begin{dfn}
  For $\Gamma_1 = S_0 \supseteq S_1 \supsetneq S_2 \supsetneq \cdots
  \supsetneq S_m \supsetneq S_{m + 1} = \varnothing$, let ${{\Gamma'}^i} $ be
  the subgraph generated by $S_i (1 \leqslant i \leqslant m)$. We call
  $C_{S_1, S_2, \ldots, S_m} \subseteq \widetilde{[0, + \infty)^{| \Gamma_1
  |}}$ the {\tmstrong{corner of compactified Schwinger space}} corresponds to
  ${\Gamma'}^1 {, \Gamma'}^2, \ldots,  {\Gamma'}^m \rightarrow 0$.
\end{dfn}

We will use $\partial_{m} C_{S_1, S_2, \ldots, S_m}$ to denote the following set:
\[ \left\{ p \in C_{S_1, S_2, \ldots, S_m} | \nobracket \rho_i (p) = 0 \text{
   for } 1 \leqslant i \leqslant m \right\} . \]
\begin{rmk}
  Note $\widetilde{[0, + \infty)^{| \Gamma_1 |}}$ is a manifold with corners.
  In particular, it is a stratified space. Its codimension m strata is
  \[ \bigcup_{\Gamma_1 = S_0 \supseteq S_1 \supsetneq \cdots \supsetneq S_m
     \supsetneq S_{m + 1} = \varnothing} \partial_{m} C_{S_1, S_2, \ldots, S_m} .
  \]
     We will denote the closure of $\partial C_{\Gamma_{1}}$ in $\widetilde{[0, + \infty)^{| \Gamma_1 |}}$ by $\widetilde{(0, + \infty)^{| \Gamma_1 |}} /\R^+$.
\end{rmk}

The matrice $(M_{\Gamma} (t)^{- 1})^{i j}$ and $(d_{\Gamma}  (t)^{- 1})^{e j}$
defined in Appendix \ref{graph theory} can be extended to smooth functions on
$\widetilde{[0, + \infty)^{| \Gamma_1 |}}$:

\begin{lem}
  \label{extended functions}Given a connected graph $\Gamma$ without
  self-loops, The following functions can be extended to smooth functions on
  $\widetilde{[0, + \infty)^{| \Gamma_1 |}}$:
  \begin{enumeratenumeric}
    \item $(M_{\Gamma} (t)^{- 1})^{i j}$ for $1 \leqslant i, j \leqslant |
    \Gamma_0 | - 1$.
    
    \item $(d_{\Gamma}  (t)^{- 1})^{e j}$ for $e \in \Gamma_1, \text{ } 1
    \leqslant j \leqslant | \Gamma_0 | - 1$.
  \end{enumeratenumeric}
\end{lem}

\begin{proof}
  See {\cite{wang2024feynman}}.
\end{proof}

The following lemma will be used in the prove of finiteness of Feynman graph
integrals for $\R^{d'}$.

\begin{lem}
  The map\label{square map}
  \[ (t_e)_{e \in \Gamma_1} \in (0, + \infty)^{| \Gamma_1 |} \rightarrow
     (t_e^2)_{e \in \Gamma_1} \in (0, + \infty)^{| \Gamma_1 |} \]
  can be extended to a smooth map $t_{\tmop{square}}$ from $\widetilde{[0, +
  \infty)^{| \Gamma_1 |}}$ to $\widetilde{[0, + \infty)^{| \Gamma_1 |}}$.
\end{lem}

\begin{proof}
  From Proposition \ref{covered by corners}, we only need to extend the map to
  a map from $C_{S_1, S_2, \ldots, S_m}$ to $C_{S_1, S_2, \ldots, S_m}$. Let's
  prove that this map can be extended to $C_{\Gamma'_1}$, where $\Gamma'_1$ is
  a subgraph of $\Gamma$. The general situation is left to the reader.
  
  This can be shown by using the coordinates $\{ \rho, t_e, \xi_{e'} \}_{e \in
  \Gamma_1 \backslash \Gamma'_1, e' \in \Gamma_1'}$ we have introduced. More
  precisely, $\{ \rho, t_e, \xi_{e'} \}_{e \in \Gamma_1 \backslash \Gamma'_1,
  e' \in \Gamma_1'}$ will be mapped to $\{ \tilde{\rho}, \widetilde{t_e},
  \widetilde{\xi_{e'}} \}_{e \in \Gamma_1 \backslash \Gamma'_1, e' \in
  \Gamma_1'}$, such that
  \[ \left\{\begin{array}{l}
       \widetilde{t_e} = t_e^2  \text{\quad} e \in \Gamma_1 \backslash
       \Gamma'_1,\\
       \widetilde{\xi_{e'}} = \frac{\xi_{e'}^2}{\sqrt{\underset{e' \in
       \Gamma'_1}{\sum} \xi_{e'}^4}} \quad e \in \Gamma_1',\\
       \tilde{\rho} = \rho^2 \sqrt{\underset{e' \in \Gamma'_1}{\sum}
       \xi_{e'}^4} .
     \end{array}\right. \]
  It is clear that this map is smooth on $C_{\Gamma'_1}$.
\end{proof}

Finally, we give a description of the boundary of compactified Schwinger
space:

\begin{prop}
  \label{boundary description}Given a connected graph $\Gamma$ without
  self-loops, the boundary $\partial \widetilde{[0, L]^{| \Gamma_1 |}}$ of
  $\widetilde{[0, L]^{| \Gamma_1 |}}$ has the following decomposition:
  \[ \partial \widetilde{[0, L]^{| \Gamma_1 |}} = \left( - \partial_0
     \widetilde{[0, L]^{| \Gamma_1 |}} \right) \cup \partial_L \widetilde{[0,
     L]^{| \Gamma_1 |}}, \]
  where $\partial_0 \widetilde{[0, L]^{| \Gamma_1 |}}$($\partial_L
  \widetilde{[0, L]^{| \Gamma_1 |}}$) describe the boundary components near
  the origin(away from the origin). More precisely, we have
  \[ \left\{\begin{array}{l}
       \partial_0 \widetilde{[0, L]^{| \Gamma_1 |}} = \bigcup_{\Gamma'
       \subseteq \Gamma} (- 1)^{\sigma (\Gamma', \Gamma / \Gamma')}
       \widetilde{(0, + \infty)^{| \Gamma_1' |}} /\R^+ \times
       \widetilde{[0, + L]^{| \Gamma_1 \backslash \Gamma_1' |}}\\
       \partial_L \widetilde{[0, L]^{| \Gamma_1 |}} = \bigcup_{e \in \Gamma_1}
       (- 1)^{| e |} \{ L \} \times  \widetilde{[0, L]^{| (\Gamma / e)_1 |}}
     \end{array}\right. . \]
\end{prop}

\begin{proof}
  This follows from Proposition \ref{covered by corners} and the explicit
  description of $C_{S_1, S_2, \ldots, S_m}$.
\end{proof}

\printbibliography

\end{document}